\documentclass[10pt]{article}
\RequirePackage[OT1]{fontenc}
\RequirePackage[amsthm,amsmath,doublespacing]{imsart}
\RequirePackage{amsthm,amsmath,latexsym,epsfig,color}
\RequirePackage{float}
\RequirePackage[colorlinks]{hyperref}
\theoremstyle{remark}
\newcommand {\apgt} {\ {\raise-.5ex\hbox{$\buildrel>\over\sim$}}\ }
\newcommand {\aplt} {\ {\raise-.5ex\hbox{$\buildrel<\over\sim$}}\ }
\usepackage{chngcntr}
\startlocaldefs
\theoremstyle{plain}
\newtheorem{theorem}{\sc Theorem}[section]
\newtheorem{lemma}{\sc Lemma}[section]
\newtheorem{corollary}{\sc Corollary}[section]
\endlocaldefs
\def\bX {{\bf X}}
\def\bR {{\bf R}}
\def\bW {{\bf W}}
\def\bR {{\bf R}}
\def\bK {{\bf K}}

\def\bR {{\bf R}}

\def\zerov{\mathbf 0}
\def\arg\min{\mathop{\mbox{\rm arg min}}}
\def\cite{\citeasnoun}
\usepackage[authoryear,numbers,round]{natbib}

\def\beqr{\begin{eqnarray}}             
\def\eeqr{\end{eqnarray}}               
\def\beq{\begin{equation}}             
\def\eeq{\end{equation}}               
\def\bc{\begin{center}}                 
\def\ec{\end{center}}                   
\def\bfh{\mathcal{H}}
\numberwithin{equation}{section}
\theoremstyle{plain}

\renewcommand{\thebibliography}[1]{
\setcounter{secnumdepth}{0}
\bigskip\bigskip \centerline{LITERATURE CITED}
\list
  {[\arabic{enumi}]}{\settowidth\labelwidth{[#1]}\leftmargin\labelwidth
    \advance\leftmargin\labelsep
    }
    \def\newblock{\hskip .11em plus .33em minus -.07em}
    \sloppy\clubpenalty4000\widowpenalty4000
    \sfcode`\.=1000\relax}

\def\Xv{\mathbf X}
\def\xv{\mathbf x}
\def\Gv{\mathbf G}

\def\yv{\mathbf y}

\def\zerov{\mathbf{0}}
\def\bv{\mathbf b}

\def\bfh{{\mathcal{H}}}

\def\beq2{\begin{equation*}}
\def\eeq2{\end{equation*}}
\def\beqr{\begin{eqnarray}}             
\def\eeqr{\end{eqnarray}}               
\def\bc{\begin{center}}                 
\def\ec{\end{center}}                   

\def\smt{{\mbox{\tiny T}}}
\def\mb{\mbox}                          
\def\smn{\sum_{i=1}^n}                  
\def\itg{\int_0^1}

\def\eps{\epsilon}
\def\la{\lambda}

\def\bfh{\mathcal H}

\def\Lc{\mathcal L}

\def\bv{\mathbf b}
\def\cv{\mathbf c}

\def\fv{\mathbf f}
\def\sv{\mathbf s}
\def\wv{\mathbf w}
\def\xv{\mathbf x}
\def\yv{\mathbf y}
\def\zv{\mathbf z}
\def\vv{\mathbf v}

\def\Gv{\mathbf G}

\def\Xv{\mathbf X}

\def\Av{\mathbf A}
\def\Bv{\mathbf B}

\def\zerov{\mathbf 0}

\def\btv{\boldsymbol{\beta}}
\def\epsv{\boldsymbol{\epsilon}}

\begin{document}
\begin{frontmatter}
\title{Sparse and Efficient Estimation for Partial Spline Models
with Increasing Dimension} \runtitle{Sparse Partial Spline}

\author{Guang Cheng,\thanksref{t1}}
\author{Hao Helen Zhang\thanksref{t2}}
\and
\author{Zuofeng Shang\thanksref{t3}}
\runauthor{Guang Cheng, Hao Helen Zhang and Zuofeng Shang}

\affiliation{Purdue University and University of Arizona}

\thankstext{t1}{Guang Cheng (Corresponding Author) is Associate Professor,
Department of Statistics, Purdue University, West Lafayette, IN
47907-2066, Email: chengg@purdue.edu. Supported by NSF Grant DMS-0906497 and CAREER Award DMS-1151692.}
\thankstext{t2}{Hao Helen Zhang is
Associate Professor, Department of Mathematics, University of Arizona, Tucson, AZ 85721-0089, Email: hzhang@math.arizona.edu. Supported
by NSF grants DMS-0645293, DMS-1347844, NIH grants P01 CA142538 and R01 CA085848.}
\thankstext{t3}{Zuofeng Shang is Visiting Assistant Professor, Department of Statistics, Purdue University, West Lafayette, IN 47907-2066, Email: shang9@purdue.edu.}

\begin{abstract}
We consider model selection and estimation for
partial spline models and propose a new regularization method in
the context of smoothing splines. The regularization method has a simple
yet elegant form, consisting of roughness penalty
on the nonparametric component and shrinkage penalty on the parametric
components, which can achieve function smoothing and sparse
estimation simultaneously. We establish the convergence rate and
oracle properties of the estimator under weak regularity
conditions. Remarkably, the estimated
parametric components are sparse and efficient, and the nonparametric component can be estimated with
the optimal rate. The procedure also has attractive computational
properties. Using the representer theory of smoothing splines, we
reformulate the objective function as a LASSO-type problem,
enabling us to use the LARS algorithm to
compute the solution path. We then extend the procedure to
situations when the number of predictors increases with the sample
size and investigate its asymptotic properties in that context.
Finite-sample performance is illustrated by simulations.\\
\end{abstract}


\end{frontmatter}

\vspace{-0.3in}{\bf Short title: Sparse Partial Spline}

\counterwithout{equation}{section}
\counterwithout{theorem}{section}
\counterwithout{lemma}{section}
\counterwithout{corollary}{section}
\vspace{-0.1in}\section{Introduction}\label{sec:intro} \vskip
-0.7cm
\subsection{Background} \vskip -0.7cm Partial smoothing splines
are an important class of semiparametric regression models.
Developed in a framework of reproducing
kernel Hilbert spaces (RKHS), these models provide a
compromise between linear and nonparametric models.

In general, a partial smoothing spline model assumes the data
$(\bX_i, T_i, Y_i)$ follow
\begin{equation}\label{plm}
Y_i=\bX_i'\btv + f(T_i) + \eps_i, \quad i=1,\cdots, n, \quad f\in
W_m[0,1],
\end{equation}
where $\bX_{i}\in R^{d}$ are linear covariates, $T_{i}\in[0,1]$ is
the nonlinear covariate, and $\eps_i$'s are independent errors
with mean zero and variance $\sigma^2$. The space $W_m[0,1]$ is
the $m^{\text{th}}$ order Sobolev Hilbert space $W_m[0,1]=\{f: f,
f^{(1)}, ..., f^{(m-1)}~ \mb{are absolutely
continuous},~f^{(m)}\in\Lc_2[0,1]\}$ for $m\geq 1$. Here $f^{(j)}$
denotes the $j$th derivative of $f$. The function $f(t)$ is the nonparametric component of the model. Denote the
observations of $(\Xv_i, T_i, Y_i)$ as $(\xv_i,t_i,y_i)$ for
$i=1,2,\ldots,n$. The standard approach to compute the partial
spline (PS) estimator is minimizing the penalized least squares:
\begin{equation}\label{par0}
(\widetilde{\boldsymbol{\beta}}_{PS}, \tilde{f}_{PS})=
\arg\min_{\btv\in R^d,f\in
W_m}\frac{1}{n}\smn\left[y_i-\xv_i^\smt\btv-f(t_i)\right]^2+
\la_1J_f^2,
\end{equation}
where $\lambda_1$ is a smoothing parameter and $J_f^2=\itg\left[f^{(m)}(t)\right]^2dt$ is the roughness penalty
on $f$; see \citet{kw71,cw79,d84,gs94} for details. It is known
that the solution $\tilde{f}_{PS}$ is a natural spline (\citet{w90}) of order
$2m-1$ on $[0,1]$ with knots at $t_i, i=1, \cdots, n$. Asymptotic theory for partial splines has been developed by
several authors \citet{SC13, r86,h86,s88,sw88}. In this paper, we
mainly consider partial smoothing splines in the framework of
\citet{w84}.

\vspace{-0.6cm} \subsection{ Model Selection for Partial Splines}
\vskip -0.7cm Variable selection is important for data analysis and model building, especially for high dimensional
data, as it helps to improve the model's prediction accuracy and interpretability. For linear
models, various penalization procedures have been proposed to obtain a sparse model, including the non-negative garrote
\citet{b95}, LASSO \citet{t96}, SCAD \citet{fl01,fp04}, and the
adaptive LASSO \citet{z06,wlj07}. Contemporary research frequently
deals with problems where the input dimension $d$ diverges to
infinity as the data sample size increases \citet{fp04}. There is
also active research going on for linear model selection in these
situations \citet{fp04, zz08, fl08,hhm08,hmz08}.

In this paper, we propose and study a new approach to variable selection for partially linear models in the
framework of smoothing splines. The procedure leads to a regularization problem in the RKHS, whose unified
formulation can facilitate numerical computation and asymptotic inferences of the estimator.
To conduct variable selection, we employ the adaptive LASSO penalty on linear parameters.
One advantage of this procedure is its easy implementation. We
show that, by using the representer theory (\citet{w90}), the optimization
problem can be reformulated as a LASSO-type problem so that the
entire solution path can be computed by the LARS
algorithm \citet{ehjt04}. We show that the new procedure can
asymptotically (i) correctly identify the sparse model structure;
(ii) estimate the nonzero $\beta_{j}$'s consistently and achieve
the semiparametric efficiency; (iii) estimate the nonparametric
component $f$ at the optimal nonparametric rate. We also
investigate the property of the new procedure with a diverging
number of predictors \citet{fp04}.

From now on, we regard $(Y_i,\bX_i)$ as i.i.d realizations from
some probability distribution. We assume
that the $\xv_{i}$'s belong to some compact subset in $R^d$, and
they are standardized such that $\smn x_{ij}/n = 0$ and $\smn
x_{ij}^2/n=1$ for $j=1,\cdots,d$, where
$\xv_i=(x_{i1},\ldots,x_{id})'$. Also assume $t_i\in [0,1]$ for
all $i$. Throughout the paper, we use the convention that $0/0=0$.
The rest of the article is organized as follows. Section 2
introduces our new double-penalty estimation procedure for partial
spline models. Section 3 is devoted to two main theoretical
results. We first establish the convergence rates and oracle
properties of the estimators in the standard situation with a
fixed $d$, and then extend these results to the situations when $d$
diverges with the sample size $n$. Section 4 gives the
computational algorithm. In particular, we show how to compute the
solution path using the LARS algorithm. The issue of parameter
tuning is also discussed. Section 5 illustrates the performance of
the procedure via simulations and real examples. Discussions and technical proofs are presented in Section 6 and 7.

\vspace{-0.6cm} \section{Method}\label{sec:pspline} \vskip -0.6cm
We assume that $0\le t_1<t_2<\cdots<t_n\le 1$. In order to achieve a smooth estimate for the nonparametric component and sparse estimates for the
parametric components simultaneously, we consider the following
regularization problem:
\begin{equation}\label{par1}
\min_{\btv\in R^d,f\in W_m}\frac{1}{n}
\smn\left[y_i-\xv_i^\smt\btv-f(t_i)\right]^2+\la_1\itg\left[f^{(m)}(t)\right]^2dt
+\la_2\sum_{j=1}^dw_j|\beta_j|.
\end{equation}
The penalty term in
\eqref{par1} is naturally formed as a combination of roughness
penalty on $f$ and the weighted LASSO penalty on $\btv$. Here, $\lambda_1$ controls the smoothness of
the estimated nonlinear function while $\la_2$ controls the
degree of shrinkage on $\beta$'s. The weight $w_j$'s are
pre-specified. For
convenience, we will refer to this procedure as PSA (the Partial
Splines with Adaptive penalty).

Note that $w_j$'s should be adaptively chosen such
that they take large values for unimportant covariates and small
values for important covariates. In particular, we propose using
$w_j=1/|\tilde{\beta}_j|^\gamma$, where
$\widetilde{\btv}=(\tilde{\beta}_1,\cdots,\tilde{\beta}_d)'$ is
some consistent estimate for $\btv$ in the model \eqref{plm}, and
$\gamma$ is a fixed positive constant. For example, the standard
partial smoothing spline $\widetilde{\boldsymbol{\beta}}_{PS}$ can
be used to construct the weights. Therefore, we get the following
optimization problem:
\begin{equation}\label{par2}
(\widehat{\btv}_{PSA},\hat{f}_{PSA})=\arg\min_{\btv\in R^d,f\in
W_m}\frac{1}{n}\sum_{i=1}^{n}\left[y_{i}-\xv_{i}'\btv-f(t_{i})
\right]^{2}+ \lambda_1\int_{0}^{1}\left[f^{(m)}(t)\right]^{2}dt+
\lambda_2\sum_{j=1}^{d}\frac{|\beta_{j}|}{|\tilde{\beta}_{j}|
^{\gamma}}.
\end{equation}

\vspace{-.1in}When $\btv$ is fixed, the standard smoothing spline
theory suggests that the solution to \eqref{par2} is linear in the
residual $(\yv-\Xv\btv)$, i.e.
$\hat{\fv}(\btv)=A(\lambda_1)(\yv-\Xv\btv),$ where
$\yv=(y_1,\ldots,y_n)'$, $\Xv=(\xv_1,\ldots,\xv_n)'$ and the
matrix $A(\lambda_1)$ is the smoother or influence matrix
\citet{w84}. The expression of $A(\lambda_1)$ will be given in
Section 4. Plugging $\hat{\fv}(\btv)$ into (\ref{par2}), we can
obtain an equivalent objective function for $\btv$:
\begin{equation}
Q(\btv)=\frac{1}{n}(\yv-\Xv\btv)'[I-A(\lambda_1)](\yv-\Xv\btv)+\lambda_2\sum_{j=1}^{d}\frac{|\beta_{j}|}{|\tilde{\beta}_{j}|^{\gamma}},
\label{par3}
\end{equation}
where $I$ is the identity matrix of size $n$.
The PSA
solution can be computed as: \vspace{-.05in}\begin{eqnarray*}
\widehat{\boldsymbol{\beta}}_{PSA}&=&\arg\min_{\btv} Q(\btv), \\
 \hat{f}_{PSA}&=&A(\lambda_1)(\yv-\Xv\widehat{\btv}_{PSA}).
\vspace{-.05in}\end{eqnarray*} Special software like Quadratic
Programming (QP) or LARS \citet{ehjt04} is needed to obtain the
solution.

\vspace{-0.6cm} \section{Statistical Theory}\label{sec:theory}
\vskip -0.6cm We can write the true
coefficient vector as
$\btv_{0}=(\beta_{01},\cdots,\beta_{0d})'=(\btv'_1,\btv'_2)'$,
where $\btv_1$ consists of all $q$ nonzero components and $\btv_2$
consists of the rest $(d-q)$ zero elements, and write the true
function of $f$ as $f_0$. We also write the estimated vector
$\widehat{\btv}_{PSA}=(\hat{\beta}_1,...,\hat{\beta}_d)=
\left(\widehat{\btv}'_{PSA,1},\widehat{\btv}'_{PSA,2}\right)'$. In addition, assume that $\Xv_i$ has zero mean and strictly positive
definite covariance matrix $\bR$. The observations $t_{i}$'s satisfy
\begin{eqnarray}
\int_{0}^{t_{i}}u(w)dw=i/n\;\;\;\mbox{for
i=1,\ldots,n},\label{tcond}
\end{eqnarray}
where $u(\cdot)$ is a continuous and strictly positive function
independent of $n$.

\vspace{-0.6cm} \subsection{Asymptotic Results for Fixed $d$}

\vskip -0.6cm We show that, for any fixed $\gamma>0$, if $\lambda_1$
and $\lambda_2$ converge to zero at proper rates, then both the
parametric and nonparametric components can be estimated at their
optimal rates. Moreover, our estimation procedure produces the
nonparametric estimate $\hat{f}_{PSA}$ with desired smoothness,
i.e. (\ref{jrate}). Meanwhile, we conclude that our
double penalization procedure can estimate the nonparametric
function well enough to achieve the oracle properties of the
weighted Lasso estimates.

In the below we use $\|\cdot\|$,
$\|\cdot\|_{2}$ to represent the Euclidean norm, $L_{2}$- norm,
and use $\|\cdot\|_{n}$ to denote the empirical $L_{2}$-norm, i.e.
$\|F\|_{n}^{2}=\sum_{i=1}^{n}F^{2}(s_{i})/n$.

We derive our convergence rate results under the following
regularity conditions: \vspace{-0.3cm} \begin{itemize} \item [R1.]
$\epsilon$ is assumed to be independent of $X$, and has a
sub-exponential tail, i.e. $E(\exp(|\epsilon|/C_{0}))\leq C_{0}$
for some $0<C_{0}<\infty$, see \citet{mvg97}; \item [R2.]
$\sum_{k}\phi_{k}\phi_{k}'/n$ converges to some non-singular
matrix with
$\phi_{k}=[1,t_{k},\cdots,t_{k}^{m-1},x_{k1},\cdots,x_{kd}]'$ in
probability.
\end{itemize}

\begin{theorem}\label{brate}
Consider the minimization problem \eqref{par2}, where $\gamma>0$
is a fixed constant. Assume the initial estimate $\tilde{\btv}$ is
consistent. If
$n^{2m/(2m+1)}\lambda_{1}\rightarrow\lambda_{10}>0$,
$\sqrt{n}\lambda_{2}\rightarrow 0$ and
\begin{eqnarray}
\frac{n^{\frac{2m-1}{2(2m+1)}}
\lambda_{2}}{|\tilde{\beta}_{j}|^{\gamma}}
\overset{P}{\longrightarrow} \lambda_{20}> 0\;\;\;\;\;\mbox{for}\;j=q+1,\ldots,d\label{smooth}
\end{eqnarray}
as $n\rightarrow\infty$, then we have
\vspace{-.3cm}\begin{enumerate} \item there exists a local
minimizer $\widehat{\btv}_{PSA}$ of (\ref{par2}) such that
\begin{equation}
\|\widehat{\btv}_{PSA}-\btv_{0}\|=O_{P}(n^{-1/2}).\label{betarate}
\end{equation}
\item the nonparametric estimate $\hat{f}_{PSA}$ satisfies
\begin{eqnarray}
\|\hat{f}_{PSA}-f_{0}\|_{n}&=&O_{P}(\lambda_{1}^{1/2}),\label{fratee}\\
J_{\hat{f}_{PSA}}&=&O_{P}(1).\label{jrate}
\end{eqnarray}
\item the local minimizer
$\widehat{\btv}_{PSA}=(\widehat{\btv}_{PSA,1}',\widehat{\btv}_{PSA,2}')'$
satisfies

\begin{enumerate}
\item[(a)] Sparsity:
$P(\widehat{\btv}_{PSA,2}=\zerov)\rightarrow 1$.

\item[(b)] Asymptotic Normality:
$$\sqrt{n}(\widehat{\btv}_{PSA,1}-\btv_1)\overset{d}{\rightarrow}
N(\zerov,\sigma^{2}\bR_{11}^{-1}),$$
where $\bR_{11}$ is the $q\times q$ upper-left sub matrix of
covariance matrix of $\Xv_i$.
\end{enumerate}
\end{enumerate}
\end{theorem}

{\it Remark.}
Note that $t$ is assumed to be nonrandom and satisfy the condition
(\ref{tcond}), and that $E\Xv=\zerov$. In this case, the
semiparametric efficiency bound for $\widehat{\btv}_{PSA,1}$ in
the partly linear model under sparsity is just
$\sigma^2\bR_{11}^{-1}$, see \citet{vw96}. Thus, we can claim that
$\widehat{\btv}_{PSA,1}$ is semiparametric efficient.\qed

If we use the partial spline solutions to construct the weights in
\eqref{par2}, and choose $\gamma=1$ and
$n^{2m/(2m+1)}\lambda_{i}\rightarrow\lambda_{i0}>0$ for $i=1,2$,
the above Theorem~\ref{brate} implies that the double penalized
estimators achieve the optimal rates for both parametric and
nonparametric estimation, i.e., (\ref{betarate})-(\ref{fratee}),
and that $\widehat{\btv}_{PSA}$ possesses the oracle properties,
i.e., the asymptotic normality of $\widehat{\btv}_{PSA,1}$ and
sparsity of $\widehat{\btv}_{PSA,2}$.

\vspace{-0.6cm} \subsection{Asymptotic Results for Diverging
$d_n$} \vskip -0.6cm Let $\btv=(\btv_{1}',\btv_{2}')'\in
R^{q_{n}}\times R^{m_{n}}=R^{d_{n}}$. Let
$\xv_{i}=(\wv_{i}',\zv_{i}')'$ where $\wv_{i}$ consists of the
first $q_{n}$ covariates, and $\zv_{i}$ consists of the remaining
$m_{n}$ covariates. Thus we can define the matrix
$\Xv_{1}=(\wv_{1},\ldots,\wv_{n})'$ and
$\Xv_{2}=(\zv_{1},\ldots,\zv_{n})'$. For any matrix $\bK$ we
denote its smallest and largest eigenvalue as $\lambda_{min}(\bK)$
and $\lambda_{max}(\bK)$, respectively.

Now, we give the additional regularity conditions
required to establish the large-sample theory for the increasing
dimensional case: \vspace{-0.3cm} \begin{enumerate} \item[R1D.]
There exist constants $0<b_{0}<b_{1}<\infty$ such that
\begin{eqnarray*}
b_{0}\leq\min\{|\beta_{j}|, 1\leq j\leq
q_{n}\}\leq\max\{|\beta_{j}|, 1\leq j\leq q_{n}\}\leq b_{1}.
\end{eqnarray*}

\item[R2D.] $\lambda_{min}(\sum_{k}\phi_{k}\phi_{k}'/n)\geq c_{3}>0$
for any $n$.

\item[R3D.] Let $\bR$ be the covariance matrix for the vector
$\Xv_i$. We assume that
\begin{eqnarray*}
0<c_{1}\leq \lambda_{min}(\bR)\leq\lambda_{max}(\bR)\leq
c_{2}<\infty\;\;\mbox{for any}\;\;n.
\end{eqnarray*}
\end{enumerate}
Conditions R2D and R3D are equivalent to Condition R2 when $d_n$ is assumed to be fixed.

\vspace{-0.5cm}\subsubsection{Convergence Rate of
$\widehat{\btv}_{PSA}$ and $\hat{f}_{PSA}$} \vspace{-0.5cm} We
first present a Lemma concerning about the convergence rate of the
initial estimate $\tilde{\btv}_{PS}$ given the increasing
dimension $d_n$. For two deterministic sequences $p_n, q_n=o(1)$,
we use the symbol $p_{n}\asymp q_{n}$ to indicate that
$p_{n}=O(q_{n})$ and $p_{n}^{-1}=O(q_{n}^{-1})$. Define $x\vee y$
($x\wedge y$) to be the maximum (minimum) value of $x$ and
$y$.\vspace{0.1in}

\begin{lemma}\label{intilemma}
Suppose that $\tilde{\btv}_{PS}$ is a partial smoothing spline
estimate, then we have \vspace{-.1in}\begin{eqnarray}
\|\tilde{\btv}_{PS}-\btv_{0}\|=O_{P}(\sqrt{d_{n}/n})\;\;
\mbox{given}\;\;d_{n}=n^{1/2}\wedge
n\lambda_{1}^{1/2m}.\label{inirate}
\end{eqnarray}
\end{lemma}\vspace{-0.1in}
Our next theorem gives the convergence rates for
$\widehat{\btv}_{PSA}$ and $\hat{f}_{PSA}$ when dimension of
$\btv_0$ diverges to infinity. In this increasing dimension
set-up, we find three results: (i) the convergence
rate for $\widehat{\btv}_{PSA}$ coincides with that for the
estimator in the linear regression model with increasing dimension
\citet{p84}, thus we can conclude that the presence of
nonparametric function and sparsity of $\btv_0$ does not affect
the overall convergence rate of $\widehat{\btv}_{PSA}$; (ii) the
convergence rate for $\hat{f}_{PSA}$ is slower than the regular
partial smoothing spline, i.e. $O_{P}(n^{-m/(2m+1)})$, and is
controlled by the dimension of important components of $\btv$,
i.e. $q_n$. (iii) the nonparametric estimator $\hat{f}_{PSA}$
always satisfies the desired smoothness condition, i.e.
$J_{\hat{f}_{PSA}}=O_P(1)$, even under increasing dimension of
$\btv$.

\begin{theorem}\label{consithm}
Suppose that $d_{n}=o(n^{1/2}\wedge n\lambda_{1}^{1/2m})$,
$n\lambda_{1}^{1/2m}\rightarrow\infty$ and
$\sqrt{n/d_n}\lambda_{2}\rightarrow 0$, we have
\begin{eqnarray}
\|\widehat{\btv}_{PSA}-\btv_{0}\|=O_{P}(\sqrt{d_{n}/n}).\label{pararatei}
\end{eqnarray}
If we further assume that $\lambda_{1}/q_{n}\asymp n^{-2m/(2m+1)}$
and
\begin{eqnarray}
\max_{j=q_{n}+1,\ldots,d_{n}}\frac{\sqrt{n/d_n}(\lambda_{2}/q_n)}{|\tilde{\beta}_{j}|^{\gamma}}=
O_{P}(n^{1/(2m+1)}d_{n}^{-3/2}), \label{l2ratei2}
\end{eqnarray}
then we have
\begin{eqnarray}
\|\widehat{f}_{PSA}-f_{0}\|_{n}&=&O_{P}(\sqrt{d_{n}/n}\vee
(n^{-m/(2m+1)}q_{n})),\label{nonpararatei}\\
J_{\hat{f}_{PSA}}&=&O_{P}(1).\label{ji}
\end{eqnarray}
\end{theorem}

It seems nontrivial to improve the rate of convergence for the parametric estimate to the minimax optimal rate $\sqrt{q_n\log{d_n}/n}$ proven in \citet{BRT09}. The main reason is that the above rate result is proven in the (finite) dictionary learning framework which requires that the nonparametric function can be well approximated by a member of the span of a finite dictionary of (basis) functions. This key assumption does not straightforwardly hold in our smoothing spline setup. In addition, it is also unclear how to relax the Gaussian error condition assumed in \citet{BRT09} to the fairly weak sub-exponential tail condition assumed in our paper.

\vspace{-0.8cm} \subsubsection{Oracle Properties} \vskip -0.6cm In
this subsection, we show that the desired oracle properties
can also be achieved even in the increasing dimension case. In
particular, when showing the asymptotic normality of
$\widehat{\btv}_{PSA,1}$, we consider an arbitrary linear
combination of $\btv_{1}$, say $\Gv_n\btv_1$, where $\Gv_n$ is an
arbitrary $l\times q_{n}$ matrix with a finite $l$.

\begin{theorem}\label{asynorthm}
Given the following conditions: \vspace{-0.3cm}\begin{enumerate}
\item[D1.] $d_{n}=o(n^{1/3}\wedge (n^{2/3}\lambda_{1}^{1/3m}))$ and
$q_{n}=o(n^{-1}\lambda_{2}^{-2})$;

\item[S1.] $\lambda_{1}$ satisfies: $\lambda_{1}/q_{n}\asymp
n^{-2m/(2m+1)}$ and $n^{m/(2m+1)}\lambda_{1}\rightarrow 0$;

\item[S2.] $\lambda_{2}$ satisfies: \vspace{-.1in}\begin{eqnarray}
\min_{j=q_{n}+1, \ldots,
d_{n}}\frac{\sqrt{n/d_{n}}\lambda_{2}}{|\tilde{\beta}_{j}|
^{\gamma}}\overset{P}{\longrightarrow}\infty,\label{l2ratei1}
\end{eqnarray}
\vspace{-.4in}\end{enumerate} we have
\vspace{-0.3cm}\begin{enumerate} \item[(a)] Sparsity:
$P(\widehat{\btv}_{PSA,2}=\zerov)\rightarrow 1$

\item[(b)] Asymptotic Normality: \vspace{-0.1in}\begin{eqnarray}
\sqrt{n}\Gv_{n}\bR_{11}^{1/2}(\widehat{\btv}_{PSA,1}-\btv_{1})
\overset{d}{\rightarrow}N(\zerov,\sigma^{2}\Gv),\label{asydisti}
\end{eqnarray}
where $\Gv_{n}$ be a non-random $l\times q_{n}$ matrix with full
row rank such that $\Gv_{n}\Gv_{n}'\rightarrow \Gv$.
\end{enumerate}
\end{theorem}

In Corollary~\ref{corid}, we give the fastest possible increasing
rates for the dimensions of $\btv_0$ and its important components
to guarantee the estimation efficiency and selection
consistency. The range of the smoothing and shrinkage
parameters are also given.

\begin{corollary}\label{corid}
Let $\gamma=1$. Suppose that $\tilde{\btv}$ is the partial
smoothing spline solution. Then, we have
\vspace{-0.2cm}\begin{enumerate} \item
$\|\widehat{\btv}_{PSA}-\btv_{0}\|=O_{P}(\sqrt{d_{n}/n})$ and
$\|\hat{f}_{PSA}-f_{0}\|_{n}=O_{P}(\sqrt{d_{n}/n}\vee
(n^{-m/(2m+1)}q_{n}))$;

\item $\widehat{\btv}_{PSA}$ possesses the oracle properties.
\end{enumerate}
\vspace{-0.3cm}if the following dimension and smoothing parameter
conditions hold:
\begin{eqnarray}
&&d_{n}=o(n^{1/3})\;\mbox{and}\;q_{n}=o(n^{1/3}),\label{dcon}\\
&&n\lambda_1^{1/2m}\rightarrow\infty,
n^{m/(2m+1)}\lambda_1\rightarrow
0\;\;\mbox{and}\;\;\lambda_1/q_n\asymp n^{-2m/(2m+1)},\label{scon1}\\
&&\sqrt{n/d_n}\lambda_2\rightarrow 0,
(n/d_n)\lambda_2\rightarrow\infty\;\;\mbox{and}\;\;\sqrt{d_n}(n/q_n)\lambda_
2=O(n^{1/(2m+1)}).\label{scon2}
\end{eqnarray}
\end{corollary}

Define $d_n\asymp n^{\widetilde d}$ and $q_n\asymp n^{\widetilde q}$,
where $0\leq \widetilde q\leq \widetilde d<1/3$ according to (\ref{dcon}).
For the usual case that $m\geq 2$, we can give a set of sufficient conditions
for (\ref{scon1})-(\ref{scon2}) as: $\lambda_1\asymp n^{-r_1}$ and $\lambda_2\asymp n^{-r_2}$
for $r_1=2m/(2m+1)-\widetilde q$, $r_2=\widetilde{d}/2+2m/(2m+1)-\widetilde q$ and
$(1-\widetilde{d})/2<r_2<1-\widetilde d$. The above conditions are very easy to check.
For example, if $m=2$, we can set $\lambda_1 \asymp n^{-0.55}$ and $\lambda_2\asymp n^{-0.675}$ when $d_n\asymp
n^{1/4}$, $q_n\asymp n^{1/4}$.

In \citet{NZZ09}, the authors considered variable selection in partly linear models when dimension is increasing.
Under $m=2$, they applied SCAD penalty on the parametric part and proved oracle property together with asymptotic normality.
In this paper, we considered general $m$ and applied adaptive LASSO for the parametric part. Besides oracle property,
we also derived rate of convergence
for the nonparametric estimate. The technical proof relies on nontrivial applications of RKHS theory and model empirical processes theory. Therefore, our results are substantially different from \citet{NZZ09}.
The numerical results provided in Section \ref{sec:example} demonstrate satisfactory selection accuracy.
Furthermore, we are able to report the estimation accuracy of the nonparametric estimate, which is also satisfactory.

\vspace{-0.5cm} \section{Computation and Tuning}

\vskip -0.6cm \subsection{Algorithm} \vskip -0.6cm
We propose a two-step procedure to obtain the
PSA estimator: first compute $\widehat{\btv}_{PSA}$, then compute $\hat{f}_{PSA}$. As shown
in Section 2, we need to minimize (\ref{par3}) to estimate $\beta$. Define the square root matrix of $I-A(\lambda_1)$ as $T$, i.e.
$I-A(\lambda_1)=T'T$. Then (\ref{par3}) can be
reformulated into a LASSO-type problem\vspace{-0.05in}
\begin{equation}
\label{lasso1}
\min\frac{1}{n}(\yv^*-\Xv^*\btv^*)'(\yv^*-\Xv^*\btv^*)+\lambda_2\sum_{j=1}^d |\beta_j^*|,
\end{equation}
where the transformed variables are $\yv^*=T\yv$, $\Xv^*=T\Xv\bW$, and $\beta_j^*=\beta_j/|\tilde{\beta}_{j}|^{\gamma}, j=1, \cdots, d$,
with $\bW=\mbox{diag}\{|\tilde{\beta}_{j}|^{\gamma}\}$. Therefore, (\ref{lasso1}) can be conveniently solved with the LARS algorithm
\citet{ehjt04}.

Now assume $\widehat{\btv}_{PSA}$ has been obtained. Using the
standard smoothing spline theory, it is easy to show that
$\hat{\fv}_{PSA}=A(\lambda_1)(\yv-\Xv\widehat{\btv}_{PSA})$, where
$A$ is the influence matrix. By the reproducing kernel Hilbert
space theory \citet{kw71}, $W_m[0,1]$ is an RKHS when equipped
with the inner product
$$(f,g)=\sum_{\nu=0}^{m-1}\left[\int_0^1f^{(\nu)}(t)dt\right]\left[\int_0^1g^{(\nu)}(t)dt\right]+\int_0^1f^{(m)}g^{(m)}dt.$$
We can decompose $W_m[0,1]=\bfh_0\oplus\bfh_1$ as a direct sum of two RKHS subspaces. In particular,
$\bfh_0=\{f:f^{(m)}=0\}=\mbox{span}\{k_\nu(t),\nu=0,\cdots,m-1\}$, where $k_\nu(t)=B_\nu(t)/\nu!$
and $B_\nu(t)$ are Bernoulli polynomials \citet{as64}.
$\bfh_1=\{f:\int_0^1f^{(\nu)}(t)dt=0,\nu=0,\cdots,m-1;f^{(m)}\in\Lc_2[0,1]\}$, associated with the
reproducing kernel $K(t,s)=k_m(t)k_m(s)+(-1)^{m-1}k_{2m}([s-t])$, where $[\tau]$ is the fractional part of
$\tau$. Let $S$ be a $n\times n$ square matrix with $s_{i,\nu}=k_{\nu-1}(t_i)$ and $\Sigma$ be a square matrix
with the $(i,j)$-th entry $K(t_i,t_j)$. Let the QR decomposition of
$S$ be $S=(F_1,F_2)\begin{pmatrix}U\\ 0\end{pmatrix}$, where $F=[F_1, F_2]$ is orthogonal and $U$ is
upper triangular with $S'F_2=0$. As shown in
\citet{w84} and \citet{g02}, the influence matrix $A$ can be expressed as
$$A(\lambda_1)=I-n\lambda_1F_2(F_2'VF_2)^{-1}F_2',$$
where $V=\Sigma+n\lambda_1I$. Using the representer theorem
(\citet{w90}), we can compute the nonparametric estimator as
$$\hat{f}_{PSA}(t)=\sum_{\nu=0}^{m-1}\hat{b}_\nu k_\nu(t)+\sum_{i=1}^n\hat{c}_iK(t, t_i),$$
where $\widehat{\cv}=F_2(F_2'VF_2)^{-1}F_2'\yv$ and $\widehat{\bv}=U^{-1}F_1'(\yv-\Sigma\widehat{\cv})$.
We summarize the algorithm in the following:

Step 1. Fit the standard smoothing spline and construct the weights $w_j$'s. Compute
$\yv^*$ and $\Xv^*$.

Step 2. Solve \eqref{lasso1} using the LARS algorithm. Denote the solution as
$\widehat{\btv}^*=(\hat{\beta}_1^*,\cdots,\hat{\beta}_d^*)'$.

Step 3. Calculate $\widehat{\btv}_{PSA}=(\hat{\beta}_1, \cdots, \hat{\beta}_d)'$ by
$\hat{\beta}_{j}=\hat{\beta}_j^*|\tilde{\beta}_{j}|^{\gamma}$ for
$j=1,\cdots, d$.

Step 4. Obtain the nonparametric fit by
$\widehat{\fv}=S\widehat{\bv}+\Sigma\widehat{\cv}$, where the
coefficients are computed as
$\widehat{\cv}=F_2(F_2'VF_2)^{-1}F_2'\yv$ and
$\widehat{\btv}=U^{-1}F_1'(\yv-\Sigma\widehat{\cv})$.
\vspace{-0.2in}

\subsection{Parameter Tuning} \vskip -0.6cm
One possible tuning approach for the double penalized estimator is
to choose $(\la_1,\la_2)$ {\it jointly} by minimizing some scores.
Following the local quadratic approximation (LQA) technique used
in \citet{t96} and \citet{fl01}, we can derive the GCV score as a
function of $(\la_1,\la_2)$. Define the diagonal matrix
$D(\btv)=\mbox{diag}\{1/|\tilde{\beta}_{1}\beta_1|,\cdots,1/|\tilde{\beta}_{d}\beta_d|\}$.
The solution $\widehat{\btv}_{PSA}$ can be approximated by
$$\left[\Xv'\{I-A(\la_1)\}\Xv+n\la_2D(\widehat{\btv}_{PSA})\right]^{-1}\Xv'\{I-A(\la_1)\}\yv
\equiv H\yv.$$ Correspondingly,
$\widehat{\fv}_{PSA}=A(\la_1)(\yv-\Xv\widehat{\btv}_{PSA})=
A(\la_1)[I-\Xv H]\yv$. Therefore, the predicted response can
be approximated as
$\widehat{\yv}=X\widehat{\btv}_{PSA}+\hat{\fv}_{PSA}=M(\la_1,\la_2)\yv$,
where
$$M(\la_1,\la_2)=\Xv H+A(\la_1)[I-\Xv H].$$
Therefore, the number of effective parameters in the double
penalized fit $(\widehat{\btv}_{PSA},\hat{\fv}_{PSA})$ may be
approximated by $\mbox{tr}\left(M(\la_1,\la_2)\right)$. The GCV
score can be constructed as
$$GCV(\la_1,\la_2)=\frac{n^{-1}\sum_{i=1}^n(y_i-\hat{y}_i)^2}{[1-n^{-1}\mbox{tr}\left(M(\la_1,\la_2)\right)]^2}.$$

The two-dimensional search is computationally expensive in
practice. In the following, we suggest an alternative
{\it two-stage} tuning procedure. Since $\la_1$ controls the
partial spline fit
$(\widetilde{\btv},\widetilde{\bv},\widetilde{\cv})$, we first
select $\la_1$ using the GCV at Step 1 of the computation
algorithm:
$$GCV(\la_1)=\frac{n^{-1}\sum_{i=1}^n(y_i-\tilde{y}_i)^2}{[1-n^{-1}\mbox{tr}
\{\tilde{A}(\la_1)\}]^2},$$ where
$\widetilde{\yv}=(\tilde{y}_1,\cdots,\tilde{y}_n)'$ is the partial
spline prediction and $\tilde{A}(\la_1)$ is the influence matrix
for the partial spline solution. Let
$\la_1^*=\mbox{arg}\min_{\la_1}\mbox{GCV}(\la_1)$. We can also
select $\la_1^*$ using GCV in the smoothing spline problem:
$Y_i-\Xv_i'\tilde{\btv}=f(t_i) + \eps_i$, where $\tilde{\btv}$ is
the $\sqrt{n}$-consistent difference-based estimator \citet{y97}.
This substitution approach is theoretically valid for selection
$\la_1$ since the convergence rate of $\tilde{\btv}$ is faster
than the nonparametric rate for estimating $f$, and thus
$\tilde{\btv}$ can be treated as the true value. At the successive
steps, we fix $\la_1$ at $\la_1^*$ and only select $\la_2$ for the
optimal variable selection. \citet{wlt07, zl07, wll09} suggested
that BIC works better in terms of consistent model selection
than the GCV when tuning $\la_2$ for the
adaptive LASSO in the context of linear models even with diverging
dimension. Therefore, we propose to choose $\la_2$ by minimizing
$$\mbox{BIC}(\la_2)=(\yv-\Xv\widehat{\btv}_{PSA}-\hat{\fv}_{PSA})'
(\yv-\Xv\widehat{\btv}_{PSA}-\hat{\fv}_{PSA})/\hat{\sigma}^2+\log(n)\cdot
r,$$ where $r$ is the number of nonzero coefficients in
$\widehat{\btv}$, and the estimated residual variance
$\hat{\sigma}^2$ can obtained from the standard partial spline
model, i.e.
$\hat{\sigma}^2=(\yv-\Xv\widetilde{\btv}_{PS}-\widetilde{\fv}_{PS})'
(\yv-\Xv\widetilde{\btv}_{PS}-\widetilde{\fv}_{PS})/(n-
\mbox{tr}(\tilde{A}(\lambda_1))-d)$.

\vspace{-0.6cm} \section{Numerical Studies}\label{sec:example}

\vspace{-0.6cm}\subsection{Simulation 1} \vspace{-0.2in}We compare the standard partial
smoothing spline model with the new procedure under the LASSO (with $w_j=1$ in (\ref{par1})) and
adaptive (ALASSO) penalty. In the following, these three methods
are respectively referred to as ``PS'', ``PSL'' and ``PSA''. We
also include the ``Oracle model'' fit assuming the true model were known. In all the examples, we use $\gamma=1$ for PSA
and consider two sample sizes $n=100$ and $n=200$. The smoothness parameter $m$ was chosen to be $2$ in all the numerical experiments.

In each setting, a total of 500 Monte Carlo (MC) simulations are carried out. We report the MC sample mean and standard deviation
(given in the parentheses) for the MSEs. Following \citet{fl04}, we use mean squared error
$MSE(\widehat{\btv})=E\|\widehat{\btv}-\btv\|^2$ and mean integrated squared error
$MISE(\hat{f})=E\left[\int_0^1\{\widehat{f}(t)-f(t)\}^2 dt\right]$ to evaluate goodness-of-fit for parametric and nonparametric
estimation, respectively, and compute them by averaging over data knots in the simulations.
To evaluate the variable selection performance of each method, we report the number of correct zero (``correct 0'') coefficients, the number of
coefficients incorrectly set to 0 (``incorrect 0''), model size, and the empirical probability of capturing the true model.

We generate data from a model $Y_i
=\Xv_i'\btv + f(T_i)+\varepsilon_i$, and consider
two following model settings: \vspace{-0.3cm}\begin{itemize} \item Model 1:
$\btv=(3,2.5,2,1.5,0,\cdots,0)'$, $d=15$ and $q=4$. And
$f_1(t)=1.5\sin(2\pi t)$.

\item Model 2: Let
$\btv=(3,\cdots,3,0,\cdots,0)'$, $d=20$ and $q=10$. The
nonparametric function $f_2(t)
=t^{10}(1-t)^4/(3B(11,5))+4t^4(1-t)^{10}/(15B(5,11))$, where the
beta function $B(u,v)=\int_0^1 t^{u-1}(1-t)^{v-1}dt$. Two model coefficient vectors $\btv_1=\btv$ and $\btv_2=\btv/3$ were considered.
The Euclidean norms of $\btv_1$ and $\btv_2$ are
$\|\btv_1\|=9.49$ and $\|\btv_2\|=3.16$ respectively. The supnorm of $f$ is $\|f\|_{\sup}=1.16$.
So the ratios $\|\btv_1\|/\|f\|_{\sup}=8.18$ and $\|\btv_2\|/\|f\|_{\sup}=2.72$,
denoted as the parametric-to-nonparametric signal ratios (PNSR).
The two settings represent high and low PNSR's respectively.
\end{itemize}
Two possible distributions for the covariates $X$ and $T$:
\vspace{-0.3cm}\begin{itemize} \item Model 1:
$X_1,\cdots,X_{15},T$ are i.i.d. generated from Unif$(0,1)$. \item
Model 2:  $\Xv=(X_1,\cdots,X_{20})'$ are standard normal with AR(1) correlation, i.e.
$\mathrm{corr}(X_i,X_j)=\rho^{|i-j|}$. $T$ follows Unif$(0,1)$ and is independent with $X_{i}$'s.
We consider $\rho=0.3$ and $\rho=0.6$.
\end{itemize}
Two possible error distributions are used in these two settings:
\vspace{-0.3cm}\begin{itemize} \item Model 1: (normal error)
$\epsilon_1\sim N(0,\sigma^2)$, with $\sigma=0.5$ and $\sigma=1$,
respectively. \item Model 2: (non-normal error) $\epsilon_2\sim
t_{10}$, $t$-distribution with degrees of freedom 10.
\end{itemize}

Table 1 compares the model fitting and variable selection
performance of various procedures in different settings for Model
1. It is evident that the PSA procedure outperforms both the PS
and PSL in terms of both the MSE and variable selection. The three
procedures give similar performance in estimating the
nonparametric function.

\begin{table}[H]
\begin{center}
\caption{Variable selection and fitting results for Model 1}
\smallskip
\begin{tabular}{|c|c|c|c|c|r|r|c|}\hline
$\sigma$ & $n$ &Method & MSE($\widehat{\btv}_{PSA}$)&
MISE($\hat{f}_{PSA}$) &\multicolumn{1}{c|}{Size} &
\multicolumn{2}{c|}{Number of Zeros}\\\hline & &&& &
&\multicolumn{1}{c}{correct 0}&\multicolumn{1}{c|}{incorrect
0}\\\hline
0.5&100& PS  & 0.578 (0.010) & 0.015 (0.000) & 15 (0) & 0 (0) & 0 (0)\\
&&PSL        & 0.316 (0.008) & 0.015 (0.000) & 7.34 (0.09) & 7.66 (0.09) & 0.00 (0.00) \\
& &PSA       & 0.234 (0.008) & 0.014 (0.000) & 4.53 (0.04) & 10.47 (0.04)& 0.00 (0.00)\\
&& Oracle    & 0.129 (0.004) & 0.014 (0.001) & 4 (0) & 11 (0) & 0 (0)\\\cline{2-8}
&200& PS     & 0.249 (0.004) & 0.008 (0.000) & 15 (0) & 0 (0) & 0 (0)\\
&&PSL        & 0.147 (0.004) & 0.008 (0.000) & 7.16 (0.09) & 7.84 (0.09) & 0.00 (0.00) \\
& &PSA       & 0.111 (0.004) & 0.008 (0.000) & 4.36 (0.04) & 10.64 (0.04)& 0.00 (0.00)\\
&& Oracle    & 0.063 (0.000) & 0.007 (0.000) & 4 (0) & 11 (0) & 0 (0)\\\hline\hline
1&100& PS    & 2.293 (0.040) & 0.055 (0.002) & 15 (0) & 0 (0) & 0 (0)\\
&&PSL        & 1.256 (0.032) & 0.051 (0.002) & 7.36 (0.09) & 7.64 (0.09) & 0.00 (0.00) \\
&&PSA        & 1.110 (0.036) & 0.051 (0.002) & 4.72 (0.05) & 10.25 (0.05)& 0.02 (0.00)\\
&& Oracle    & 0.511 (0.017) & 0.048 (0.002) & 4 (0) & 11 (0) & 0 (0)\\\cline{2-8}
&200& PS     & 0.989 (0.017) & 0.028 (0.001) & 15 (0) & 0 (0) & 0 (0)\\
&&PSL        & 0.587 (0.016) & 0.027 (0.001) & 7.20 (0.09) & 7.80 (0.09) & 0.00 (0.00) \\
& &PSA       & 0.479 (0.014) & 0.026 (0.001) & 4.42 (0.04) & 10.58 (0.04)& 0.00 (0.00)\\
&& Oracle    & 0.252 (0.008) & 0.026 (0.001) & 4 (0) & 11 (0) & 0\\ \hline
\end{tabular}
\end{center}
\smallskip
\end{table}

\begin{table}[H]
\begin{center}
\caption{Variable selection relative frequency in percentage over 500 runs for Model 1}
\smallskip
\begin{tabular}{|c|c|c|cc|rrrrrrrrrrr|c|}\hline
$\sigma$ & $n$& &\multicolumn{2}{c|}{important index} &
\multicolumn{11}{c|}{unimportant variable index} & P(correct)\\
\hline &&& $1-3$ &$4$ & $5$ & $6$ & $7$& $8$ & $9$ & $10$& $11$ &
$12$& $13$ & $14$ & $15$& \\ \hline
0.5&100& PSL &1&1&0.30&0.29&0.32&0.33& 0.30& 0.25& 0.32& 0.29& 0.32& 0.33& 0.30 & 0.09\\
&   & PSA &1&1&0.05& 0.05 &0.07& 0.06& 0.06& 0.03& 0.04& 0.03& 0.05& 0.06& 0.04 &0.70\\\cline{2-17}
&200& PSL &1&1&0.29& 0.29& 0.26& 0.31& 0.31& 0.28& 0.25& 0.30& 0.28& 0.29& 0.30& 0.10\\
&   & PSA &1&1&0.03& 0.03& 0.04& 0.04& 0.03& 0.03& 0.03& 0.04& 0.03& 0.03& 0.04& 0.78\\\hline\hline
1&100& PSL &1&1&0.30& 0.30& 0.32& 0.33& 0.31& 0.26& 0.32& 0.29& 0.32& 0.33& 0.31 & 0.08\\
&   & PSA &1&0.98&0.08& 0.07& 0.09& 0.08& 0.08& 0.04& 0.05& 0.05& 0.07& 0.08& 0.05 &0.55\\\cline{2-17}
&200& PSL &1&1&0.29& 0.29& 0.27& 0.32& 0.31& 0.29& 0.24& 0.31& 0.29& 0.29& 0.29 & 0.10\\
&   & PSA &1&1&0.03& 0.03& 0.05& 0.05& 0.04& 0.03& 0.03& 0.04& 0.04& 0.03& 0.04 &0.75\\\hline
\end{tabular}
\end{center}
\end{table}

Table 2 shows that the PSA works much better in distinguishing important variables from unimportant
variables than PSL. For example, when $\sigma=0.5$, the PSA
identifies the correct model $500\times 0.70=350$ times out of 500 times when
$n=100$ and $500\times 0.78=390$ times when $n=200$, while the PSL identifies the
correct model only $500\times 0.09=45$ times when $n=100$ and $500\times 0.10=50$ times when
$n=200$.

To present the performance of our nonparametric estimation
procedure, we plot the estimated functions for Model 1 in the
below Figure 1. The top row of Figure 1 depicts the typical
estimated curves corresponding to the $10$th best, the $50$th best
(median), and the $90$th best according to MISE among 100
simulations when $n=200$ and $\sigma=0.5$. It can be seen that the
fitted curves are overall able to capture the shape of the true
function very well. In order to describe the sampling variability
of the estimated nonparametric function at each point, we also
depict a $95\%$ pointwise confidence interval for $f$ in the
bottom row of Figure 1. The upper and lower bound of the
confidence interval are respectively given by the $2.5$th and
$97.5$th percentiles of the estimated function at each grid point
among 100 simulations. The results show that the function $f$ is
estimated with very good accuracy.

\begin{figure}[H]
\centering{
\includegraphics[angle=0,scale=0.7]{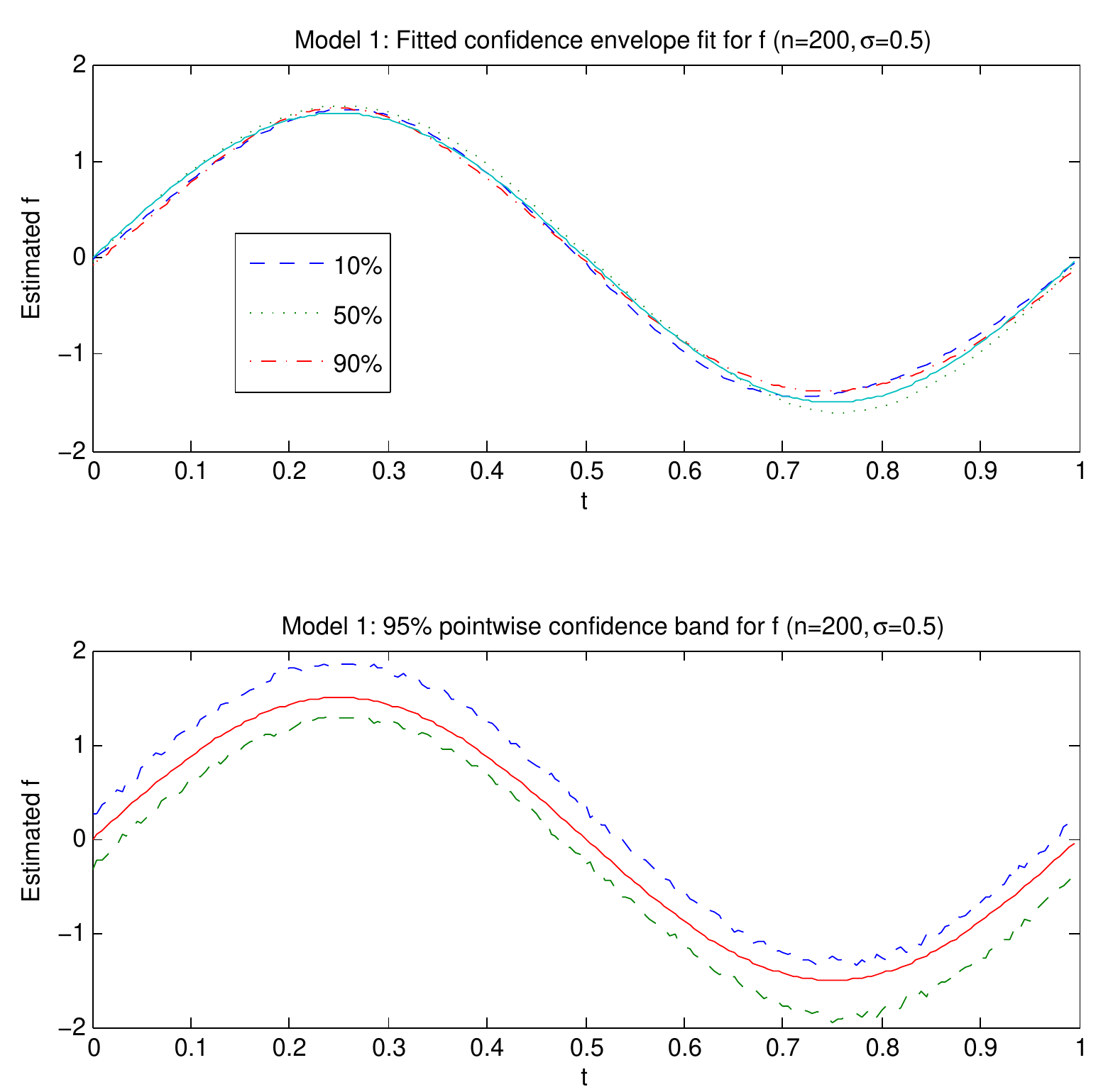}}
\caption{The estimated nonlinear functions given by the PSA in Model 1.}
\label{Ex1.plot}
{\small The estimated nonlinear function,
confidence envelop and 95\% point-wise confidence interval for
Model 2 with $n=200$ and $\sigma=0.5$. In the top plot, the dashed
line is for the $10$th best fit, the dotted line is for the $50$th
best fit, and the dashed-dotted line is for the $90$th best among
500 simulations. The bottom plot is a $95\%$ pointwise confidence
interval.}
\end{figure}

Tables 3 and 4, Tables 5 and 6 summarize the simulation results when the true parametric components are $\btv_1$
and $\btv_2$ respectively.
Tables 3 and 5 compare the model fitting and variable selection
performance in the correlated setting Model 2. The case $\rho=0.3$ represents a weak correlation among
$X$'s and $\rho=0.6$ represents a moderate situation. Again, we
observe that the PSA performs best in terms of both MSE and
variable selection in all settings. In particular, when $n=200$,
the PSA is very close to the ``Oracle'' results in this example.

Tables 4 and 6 compare the variable selection results of PSL and PSA in
four scenarios if the covariates are correlated. Since neither of
the methods misses any important variable over 500 runs, we only
report the selection relative frequencies for the unimportant variables. Overall,
the PSA results in a more sparse model and identifies the true
model with a much higher frequency. For example, when the true parametric component is $\btv_1$,
$n=100$ and
the correlation is moderate with $\rho=0.6$, the PSA identifies the
correct model with relative frequency $0.87$ (about $500\times 0.87=435$ times) while the PSL identifies
the correct model only $500\times 0.20=100$ times.
When the true parametric component is $\btv_2$,
PSA and PSL identify the correct model 405 and 90 times respectively.

\begin{table}[H]
\begin{center}
\caption{Model selection and fitting results for Model 2 when the true parameter vector is $\btv_1=\btv$.
PNSR $\approx 8.18$.}
\smallskip
\begin{tabular}{|c|c|c|c|c|r|c|c|}\hline
$\rho$ & $n$ &Method & MSE($\widehat{\btv}_{PSA})$)&
MISE($\hat{f}_{PSA}$) &\multicolumn{1}{c|}{Size} &
\multicolumn{2}{c|}{Number of Zeros}\\\hline & &&& &
&\multicolumn{1}{c}{correct 0}&\multicolumn{1}{c|}{incorrect
0}\\\hline
0.3&100& PS  & 0.416 (0.008) & 0.451 (0.002) & 20 (0) & 0 (0) & 0 (0)\\
&&PSL        & 0.299 (0.006) & 0.447 (0.002) & 12.99 (0.09) & 7.01 (0.09) & 0.00 (0.00) \\
&&PSA        & 0.204 (0.005) & 0.443 (0.002) & 10.29 (0.03) & 9.71 (0.03)& 0.00 (0.00)\\
&& Oracle    & 0.181 (0.004) & 0.444 (0.002) & 10 (0) & 10 (0) & 0
(0)\\\cline{2-8}
&200& PS     & 0.179 (0.003) & 0.408 (0.001) & 20 (0) & 0 (0) & 0 (0)\\
&&PSL        & 0.125 (0.003) & 0.406 (0.001) & 13.22 (0.08) & 6.78 (0.8) & 0.00 (0.00) \\
& &PSA       & 0.087 (0.002) & 0.404 (0.001) & 10.12 (0.02) & 9.88 (0.02)& 0.00 (0.00)\\
&& Oracle    & 0.082 (0.002) & 0.405 (0.001) & 10 (0) & 10 (0) & 0
(0)\\ \hline\hline
0.6&100& PS  & 0.721 (0.013) & 0.448 (0.002) & 20 (0) & 0 (0) & 0 (0)\\
&&PSL        & 0.401 (0.009) & 0.440 (0.002) & 11.91 (0.07) & 8.09 (0.07) & 0.00 (0.00) \\
& &PSA       & 0.349 (0.008) & 0.438 (0.002) & 10.22 (0.03) & 9.78 (0.03)& 0.00 (0.00)\\
&& Oracle    & 0.310 (0.004) & 0.439 (0.002) & 10 (0) & 10 (0) & 0
(0)\\\cline{2-8}
&200& PS     & 0.311 (0.005) & 0.408 (0.001) & 20 (0) & 0 (0) & 0 (0)\\
&&PSL        & 0.170 (0.004) & 0.405 (0.001) & 12.47 (0.07) & 7.53 (0.07) & 0.00 (0.00) \\
& &PSA       & 0.147 (0.004) & 0.404 (0.001) & 10.12 (0.02) & 9.88 (0.02)& 0.00 (0.00)\\
&& Oracle    & 0.139 (0.004) & 0.405 (0.001) & 10 (0) & 10 (0) & 0
(0)\\ \hline\hline
\end{tabular}
\end{center}
\smallskip
\end{table}

\begin{table}[H]
\begin{center}
\caption{Relative frequency of variables selected in 500 runs for Model 2 when the true parameter vector is $\btv_1=\btv$.
}
\smallskip
\begin{tabular}{|c|c|c|rrrrrrrrrr|c|}\hline
$\rho$& $n$&Method & \multicolumn{10}{c|}{unimportant variable index}& P(correct)
\\ \hline &&  & $11$ &12 & $13$ & $14$ & $15$ &
$16$ & $17$& $18$ & $19$ & $20$ & \\ \hline
0.3&100& PSL&0.35&0.34&0.34&0.34&0.35&0.33&0.34&0.35&0.34&0.36 & 0.11\\
  & & PSA &0.05&0.03&0.03&0.03&0.04&0.03&0.04&0.04&0.02&0.04&0.82 \\\cline{2-14}
&200& PSL &0.34&0.32&0.30&0.32&0.29&0.30&0.30&0.31&0.30&0.30& 0.14\\
 &  & PSA &0.02&0.01&0.01&0.01&0.01&0.01&0.01&0.01&0.01&0.01&0.91\\\hline\hline
0.6&100& PSL &0.32&0.26&0.24&0.25&0.25&0.23&0.22&0.25&0.24&0.26& 0.20\\
&   & PSA &0.06&0.02&0.01&0.02&0.02&0.01&0.02&0.02&0.02&0.01 &0.87\\\cline{2-14}
&200& PSL&0.29&0.23&0.20&0.18&0.19&0.20&0.18&0.17&0.19&0.20 & 0.25\\
 &  & PSA &0.02&0.01&0.01&0.00&0.00&0.00&0.01&0.00&0.00&0.00&0.96\\\hline
\end{tabular}
\end{center}
\end{table}

\begin{table}[H]
\begin{center}
\caption{Model selection and fitting results for Model 2 when the true parameter vector is $\btv_2=\btv/3$.
PNSR $\approx 2.72$.}
\smallskip
\begin{tabular}{|c|c|c|c|c|r|c|c|}\hline
$\rho$ & $n$ &Method & MSE($\widehat{\btv}_{PSA})$)&
MISE($\hat{f}_{PSA}$) &\multicolumn{1}{c|}{Size} &
\multicolumn{2}{c|}{Number of Zeros}\\\hline & &&& &
&\multicolumn{1}{c}{correct 0}&\multicolumn{1}{c|}{incorrect
0}\\\hline
0.3&100& PS  & 0.520 (0.008) & 0.432 (0.002) & 20 (0) & 0 (0) & 0 (0)\\
&&PSL        & 0.321 (0.005) & 0.430 (0.002) & 13.32 (0.08) & 6.68 (0.08) & 0.00 (0.00) \\
&&PSA        & 0.269 (0.005) & 0.425 (0.002) & 10.44 (0.05) & 9.56 (0.05)& 0.00 (0.00)\\
&& Oracle    & 0.230 (0.004) & 0.422 (0.002) & 10 (0) & 10 (0) & 0
(0)\\\cline{2-8}
&200& PS     & 0.226 (0.003) & 0.390 (0.001) & 20 (0) & 0 (0) & 0 (0)\\
&&PSL        & 0.141 (0.003) & 0.390 (0.001) & 13.11 (0.08) & 6.89 (0.08) & 0.00 (0.00) \\
& &PSA       & 0.116 (0.002) & 0.387 (0.001) & 10.38 (0.04) & 9.62 (0.04)& 0.00 (0.00)\\
&& Oracle    & 0.110 (0.002) & 0.386 (0.001) & 10 (0) & 10 (0) & 0
(0)\\ \hline\hline
0.6&100& PS  & 0.901 (0.013) & 0.432 (0.002) & 20 (0) & 0 (0) & 0 (0)\\
&&PSL        & 0.420 (0.007) & 0.425 (0.002) & 12.50 (0.08) & 7.95 (0.08) & 0.00 (0.00) \\
& &PSA       & 0.382 (0.007) & 0.419 (0.002) & 10.30 (0.04) & 9.70 (0.04)& 0.00 (0.00)\\
&& Oracle    & 0.350 (0.005) & 0.407 (0.002) & 10 (0) & 10 (0) & 0
(0)\\\cline{2-8}
&200& PS     & 0.382 (0.005) & 0.388 (0.001) & 20 (0) & 0 (0) & 0 (0)\\
&&PSL        & 0.192 (0.003) & 0.387 (0.001) & 12.20 (0.07) & 7.80 (0.07) & 0.00 (0.00) \\
& &PSA       & 0.181 (0.004) & 0.386 (0.001) & 10.17 (0.02) & 9.83 (0.02)& 0.00 (0.00)\\
&& Oracle    & 0.174 (0.003) & 0.385 (0.001) & 10 (0) & 10 (0) & 0
(0)\\ \hline\hline
\end{tabular}
\end{center}
\smallskip
\end{table}

\begin{table}[H]
\begin{center}
\caption{Relative frequency of variables selected in 500 runs for Model 2 when the true parameter vector is $\btv_2=\btv/3$.
}
\smallskip
\begin{tabular}{|c|c|c|rrrrrrrrrr|c|}\hline
$\rho$& $n$&Method & \multicolumn{10}{c|}{unimportant variable index}& P(correct)
\\ \hline &&  & $11$ &12 & $13$ & $14$ & $15$ &
$16$ & $17$& $18$ & $19$ & $20$ & \\ \hline
0.3&100& PSL&0.36&0.33&0.35&0.34&0.37&0.33&0.34&0.35&0.32&0.37 & 0.09\\
  & & PSA &0.06&0.04&0.05&0.05&0.03&0.04&0.05&0.05&0.05&0.05&0.78 \\\cline{2-14}
&200& PSL &0.34&0.32&0.31&0.35&0.32&0.32&0.32&0.30&0.31&0.33& 0.11\\
 &  & PSA &0.04&0.04&0.03&0.02&0.03&0.03&0.03&0.03&0.02&0.03&0.87\\\hline\hline
0.6&100& PSL &0.34&0.28&0.25&0.23&0.25&0.23&0.24&0.25&0.24&0.28& 0.18\\
&   & PSA &0.08&0.03&0.03&0.03&0.03&0.03&0.02&0.03&0.03&0.03 &0.81\\\cline{2-14}
&200& PSL&0.29&0.25&0.21&0.15&0.19&0.20&0.16&0.18&0.19&0.19 & 0.22\\
 &  & PSA &0.04&0.02&0.02&0.02&0.01&0.01&0.01&0.01&0.02&0.01&0.92\\\hline
\end{tabular}
\end{center}
\end{table}

The top rows of Figures 2 and 3 depict the typical estimated functions
corresponding to the $10$th best, the $50$th best (median), and
the $90$th best fits according to MISE among 100 simulations when
$n=200$, $\rho=0.3$, and the true Euclidean parameters are $\btv_1$ and $\btv_2$ respectively.
It is evident that the estimated curves
are able to capture the shape of the true function very well. The
bottom rows of Figures 2 and 3 depict the $95\%$ pointwise confidence
intervals for $f$. The results show that, when the true Euclidean parameters are
$\btv_1$ and $\btv_2$ respectively, the function $f$ is
estimated with reasonably good accuracy. Interestingly, in this simulation setting,
the choice of $\btv_1$ and $\btv_2$ does not affect much on the estimation accuracy of $f$.

\begin{figure}[H]
\centering{
\includegraphics[angle=0,scale=0.7]{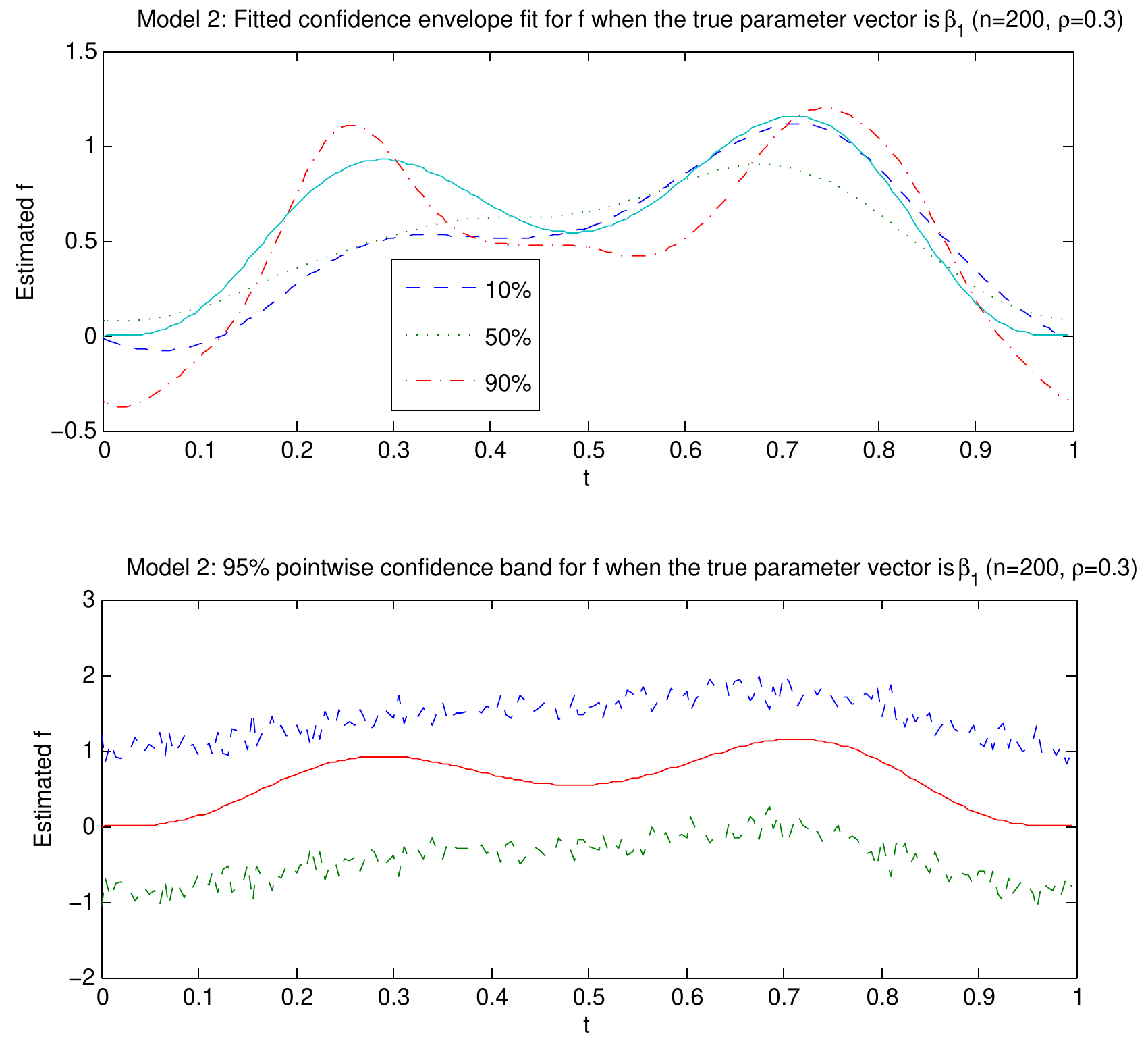}}
\caption{The estimated nonlinear functions given by the PSA in Model 2.}
\label{Ex2.plot}
{\small The estimated nonlinear function,
confidence envelop and 95\% point-wise confidence interval for
Model 2, $n=200$, $\rho=0.3$, and the true Euclidean parameter is $\btv_1$. In the top plot, the dashed
line is $10$th best fit, the dotted line is $50$th
best fit, and the dashed-dotted line is $90$th best of
500 simulations. The bottom plot is a $95\%$ pointwise confidence
interval.}
\end{figure}

\begin{figure}[H]
\centering{
\includegraphics[angle=0,scale=0.7]{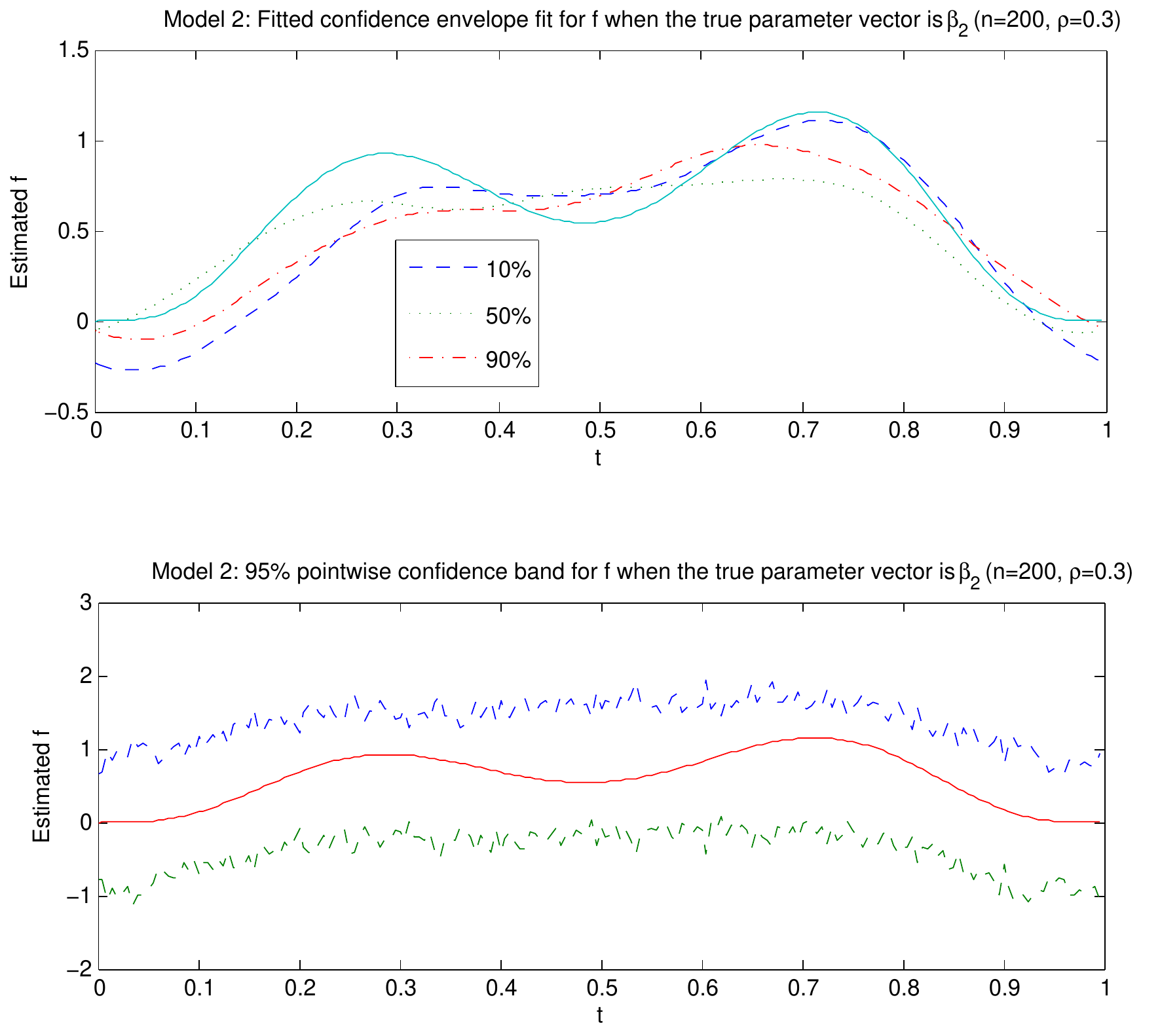}}
\caption{The estimated nonlinear functions given by the PSA in Model 2.}
\label{Ex2.plot}
{\small The estimated nonlinear function,
confidence envelop and 95\% point-wise confidence interval for
Model 2, $n=200$, $\rho=0.3$, and the true Euclidean parameter is $\btv_2$. In the top plot, the dashed
line is $10$th best fit, the dotted line is $50$th
best fit, and the dashed-dotted line is $90$th best of
500 simulations. The bottom plot is a $95\%$ pointwise confidence
interval.}
\end{figure}

\vspace{-0.6cm}\subsection{Simulation 2: Large dimensional setting} \vspace{-0.2in}

We consider an example involving a larger number of linear
variables: \vspace{-0.3cm}\begin{itemize} \item Model 3:
\vspace{0.2mm} Let $d=60, q=15$.
We considered two parameter vectors $\btv_1=\btv$ and $\btv_2=0.3\btv$,
and two nonparametric functions $f_1(t)=f(t)$ and $f_2(t)=0.5f(t)$,
with different magnitudes on the (non)parametric component representing ``weak" and ``strong" (non)parametric signals,
where $\btv=(4,4,4,4,4,3,3,3,3,3,2,2,2,2,2,0,\cdots,0)'$ and
$f(t)=0.2 t^{29}(1-t)^{16}/B(30,17)+0.8 t^{2}(1-t)^{10}/B(3,11)$.
In particular, the maximum absolute values of $f_1$ and $f_2$ are $\|f_1\|_{\sup}=3.08$ and $\|f_2\|_{\sup}=1.54$
respectively, and the $\ell_2$-norms of the $\btv_1$ and $\btv_2$ are $\|\btv_1\|=12.04$ and $\|\btv_2\|=3.61$ respectively.
So the ratio of $\|\btv_2\|$ to $\|f_1\|_{\sup}$ and $\|\btv_1\|$ to $\|f_2\|_{\sup}$,
i.e., the PNSR's,
are $\|\btv_2\|/\|f_1\|_{\sup}=1.17$
and $\|\btv_1\|/\|f_2\|_{\sup}=7.82$, representing the lower to higher PNSRs.
The correlated covariates
$(X_1,\cdots,X_{60})'$ are generated from marginally standard
normal with AR(1) correlation with $\rho=0.5$. Consider two
settings for the normal error $\epsilon_1\sim N(0,\sigma^2)$, with
$\sigma=0.5$ and $\sigma=1.5$, respectively.
\end{itemize}

\begin{table}[H]
\begin{center}
\caption{Variable selection and fitting results for Model 3 when the true parameter vector is
$\btv_2$ and the true function is $f_1$. PNSR $\approx 1.17$. SNRs $\approx 7.22$ and $2.41$ for $\sigma=0.5$ and $1.5$ respectively.}
\smallskip
\begin{tabular}{|c|c|c|c|c|r|r|c|}\hline
$\sigma$ & $n$ &Method & MSE($\widehat{\btv}$)& MISE($\hat{f}$) &\multicolumn{1}{c|}{Size}
& \multicolumn{2}{c|}{Number of Zeros}\\\hline
& &&& &
&\multicolumn{1}{c}{correct 0}&\multicolumn{1}{c|}{incorrect
0}\\\hline
0.5&100& PS     & 4.194 (0.059)& 1.139 (0.020)& 60 (0) & 0 (0) & 0 (0)  \\
   &   & PSL    & 0.526 (0.017)& 0.641 (0.011)& 27.86 (0.30) & 32.14 (0.30) & 0.00 (0.00) \\
   &   & PSA    & 0.304 (0.011)& 0.601 (0.011)& 21.43 (0.29) & 38.56 (0.29) & 0.00 (0.00) \\
   &   & Oracle & 0.225 (0.010)& 0.522 (0.011)& 15 (0) & 45 (0) & 0 (0)\\\cline{2-8}
   &200& PS     & 0.905 (0.010)& 0.851 (0.013)& 60 (0) & 0 (0) & 0 (0) \\
   &   & PSL    & 0.223 (0.007)& 0.618 (0.011)& 26.12 (0.22) & 33.88 (0.22) & 0.00 (0.00) \\
   &   & PSA    & 0.134 (0.004)& 0.548 (0.010)& 18.78 (0.20) & 41.22 (0.20) & 0.00 (0.00) \\
   &   & Oracle & 0.102 (0.002)& 0.478 (0.009)& 15 (0) & 45 (0) & 0 (0)\\\hline\hline
1.5&100& PS     & 10.014 (0.131)& 1.500 (0.027)& 60 (0) & 0 (0) & 0 (0) \\
   &   & PSL    & 2.702 (0.066) &1.256 (0.014)& 27.90 (0.35) & 32.10 (0.35) & 0.00 (0.00)\\
   &   & PSA    & 1.410 (0.040) &1.220 (0.012)& 21.70 (0.22) & 38.30 (0.22) & 0.00 (0.00)\\
   &   & Oracle & 1.038 (0.015) &1.128 (0.009)& 15 (0) & 45 (0) & 0 (0)\\\cline{2-8}
   &200& PS     & 2.440 (0.020) &1.091 (0.030)& 60 (0) & 0 (0) & 0 (0)\\
   &   & PSL    & 0.688 (0.011) &1.063 (0.003)& 26.00 (0.25) & 35.00 (0.25) & 0.00 (0.00)\\
   &   & PSA    & 0.471 (0.008) &1.052 (0.002)& 21.05 (0.20) & 38.95 (0.20) & 0.00 (0.00)\\
   &   & Oracle & 0.448 (0.006) &1.042 (0.002)& 15 (0) & 45 (0) & 0 (0)\\ \hline
\end{tabular}
\end{center}
\smallskip
\end{table}

\begin{table}[H]
\begin{center}
\caption{Variable selection and fitting results for Model 3 when the true parameter vector is
$\btv_1$ and the true function is $f_2$. PNSR $\approx 7.82$. SNRs $\approx 24.08$ and $8.03$ for $\sigma=0.5$ and $1.5$ respectively.}
\smallskip
\begin{tabular}{|c|c|c|c|c|r|r|c|}\hline
$\sigma$ & $n$ &Method & MSE($\widehat{\btv}_{PSA}$)& MISE($\hat{f}_{PSA}$) &\multicolumn{1}{c|}{Size} & \multicolumn{2}{c|}{Number of Zeros}\\\hline
& &&& & &\multicolumn{1}{c}{correct 0}&\multicolumn{1}{c|}{incorrect
0}\\\hline
0.5&100& PS     & 1.599 (0.034) & 0.321 (0.003) & 60 (0) & 0 (0) & 0 (0)\\
   &   & PSL    & 0.334 (0.019) & 0.241 (0.002) & 18.44 (0.06) & 41.56 (0.06) & 0.00 (0.00)\\
   &   & PSA    & 0.203 (0.011) & 0.232 (0.002) & 15.30 (0.03) & 44.70 (0.06) & 0.00 (0.00)\\
   &   & Oracle & 0.163 (0.002) & 0.220 (0.002) & 15 (0) & 45 (0) & 0 (0)\\\cline{2-8}
   &200& PS     & 0.379 (0.003) & 0.254 (0.001) & 60 (0) & 0 (0) & 0 (0)\\
   &   & PSL    & 0.112 (0.002) & 0.236 (0.001) & 17.19 (0.03) & 42.81 (0.03) & 0.00 (0.00)\\
   &   & PSA    & 0.094 (0.001) & 0.230 (0.001) & 15.02 (0.01) & 44.98 (0.01) & 0.00 (0.00)\\
   &   & Oracle & 0.068 (0.000) & 0.208 (0.001) & 15 (0) & 45 (0) & 0 (0) \\\hline\hline
1.5&100& PS     & 7.271 (0.083) & 0.539 (0.013) & 60 (0) & 0 (0) & 0 (0)\\
   &   & PSL    & 1.588 (0.035) & 0.459 (0.009) & 27.69 (0.22) & 32.31 (0.22) & 0.00 (0.00) \\
   &   & PSA    & 1.285 (0.028) & 0.434 (0.007) & 21.12 (0.19) & 38.88 (0.19) & 0.00 (0.00)\\
   &   & Oracle & 0.785 (0.011) & 0.395 (0.004) & 15 (0) & 45 (0) & 0 (0)\\\cline{2-8}
   &200& PS     & 1.886 (0.014) & 0.351 (0.003) & 60 (0) & 0 (0) & 0 (0)\\
   &   & PSL    & 0.571 (0.010) & 0.339 (0.002) & 26.35 (0.19) & 33.65 (0.19) & 0.00 (0.00)\\
   &   & PSA    & 0.472 (0.006) & 0.334 (0.002) & 19.08 (0.14) & 40.92 (0.14) & 0.00 (0.00)\\
   &   & Oracle & 0.342 (0.005) & 0.325 (0.002) & 15 (0) & 45 (0) & 0 (0)\\ \hline
\end{tabular}
\end{center}
\smallskip
\end{table}

Tables 7 and 8 compare the model fitting and variable selection
performance of various procedures in different settings for Model 3. In particular,
in Table 7, the true Euclidean parameter and nonparametric function are $\btv_2$ and $f_1$,
while in Table 8 they are $\btv_1$ and $f_2$.
So the PNSR's in Tables 7 and 8 are 1.17 and 7.82 respectively. To better illustrate the performance,
we considered $\sigma=0.5$ and $1.5$ in each table. The corresponding signal-to-noise ratios, defined as
the ratios of $\|\btv_1\|$ ($\|\btv_2\|$) to $\sigma$'s, are 24.08, 8.03, 7.22 and 2.41 in the four settings.
It is evident that the PSA procedure outperforms both the PS
and PSL in terms of both the MSE and variable selection accuracy. The three procedures give similar performance in estimating the
nonparametric function. In Figures 4 and 5,
we plotted the confidence envelop and the 95\% confidence band of $f_1$ and $f_2$ when $n=200$, $\sigma=0.5$,
and the true Euclidean parameters are $\btv_2$ and $\btv_1$ respectively.
All the figures demonstrate satisfactory coverage of the true unknown nonparametric
function by confidence envelops and pointwise confidence bands. We also conclude that,
at least in this simulation setup, for the two settings with different PNSR's,
the estimates of $f_1$ and $f_2$ are satisfactory.

\begin{figure}[H]
\centering{
\includegraphics[angle=0,scale=0.7]{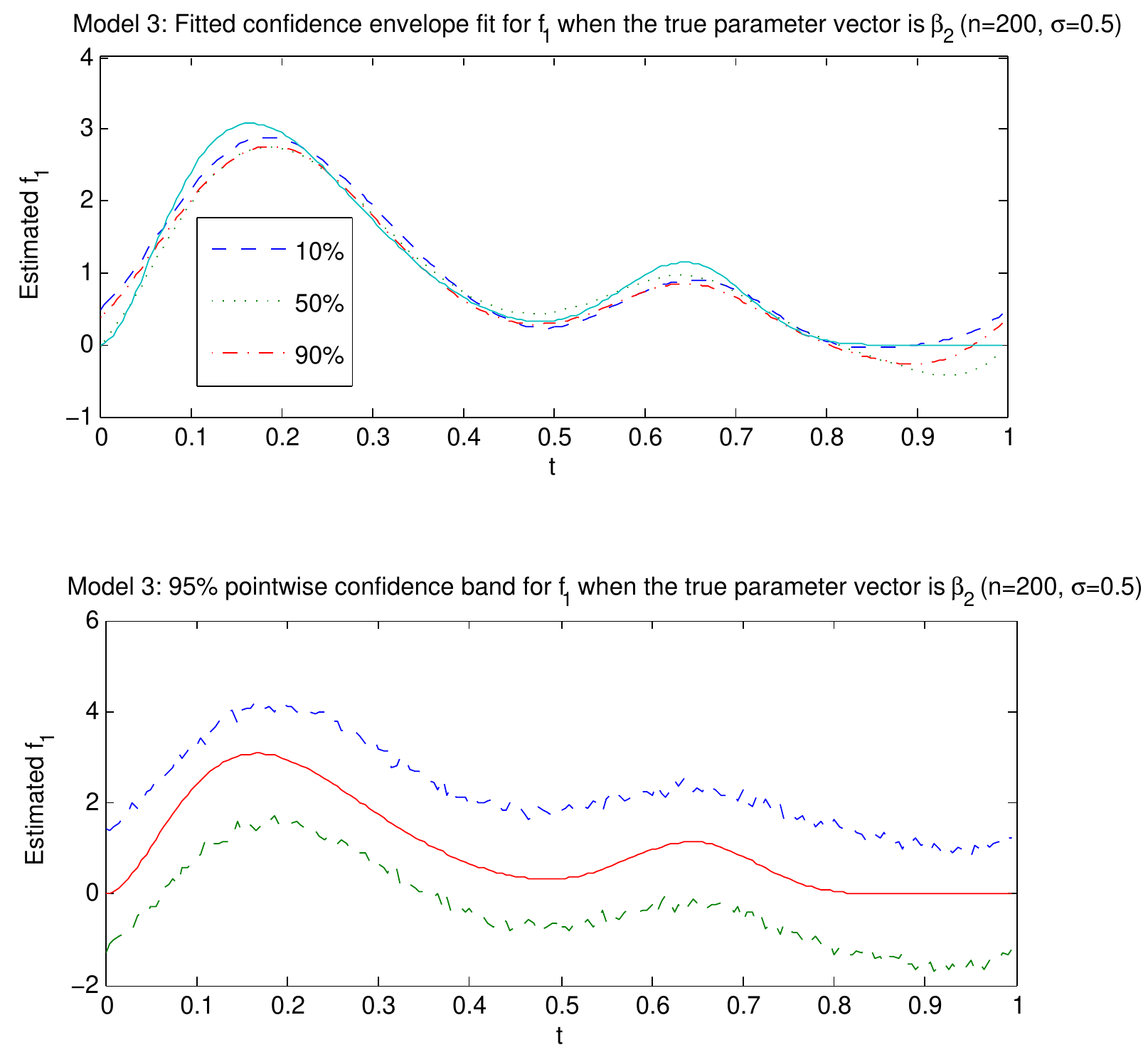}}
\caption{The estimated nonlinear functions given by the PSA in Model 3.}
\label{Ex3.plot1}
{\small The estimated nonlinear function,
confidence envelop and 95\% point-wise confidence interval for
Model 3 with true nonparametric function $f_1$ and true Euclidean parameter $\btv_2$,
$n=200$ and $\sigma=0.5$. In the top plot, the dashed
line is for the $10$th best fit, the dotted line is for the $50$th
best fit, and the dashed-dotted line is for the $90$th best among
500 simulations. The bottom plot is a $95\%$ pointwise confidence
interval.}
\end{figure}

\begin{figure}[H]
\centering{
\includegraphics[angle=0,scale=0.7]{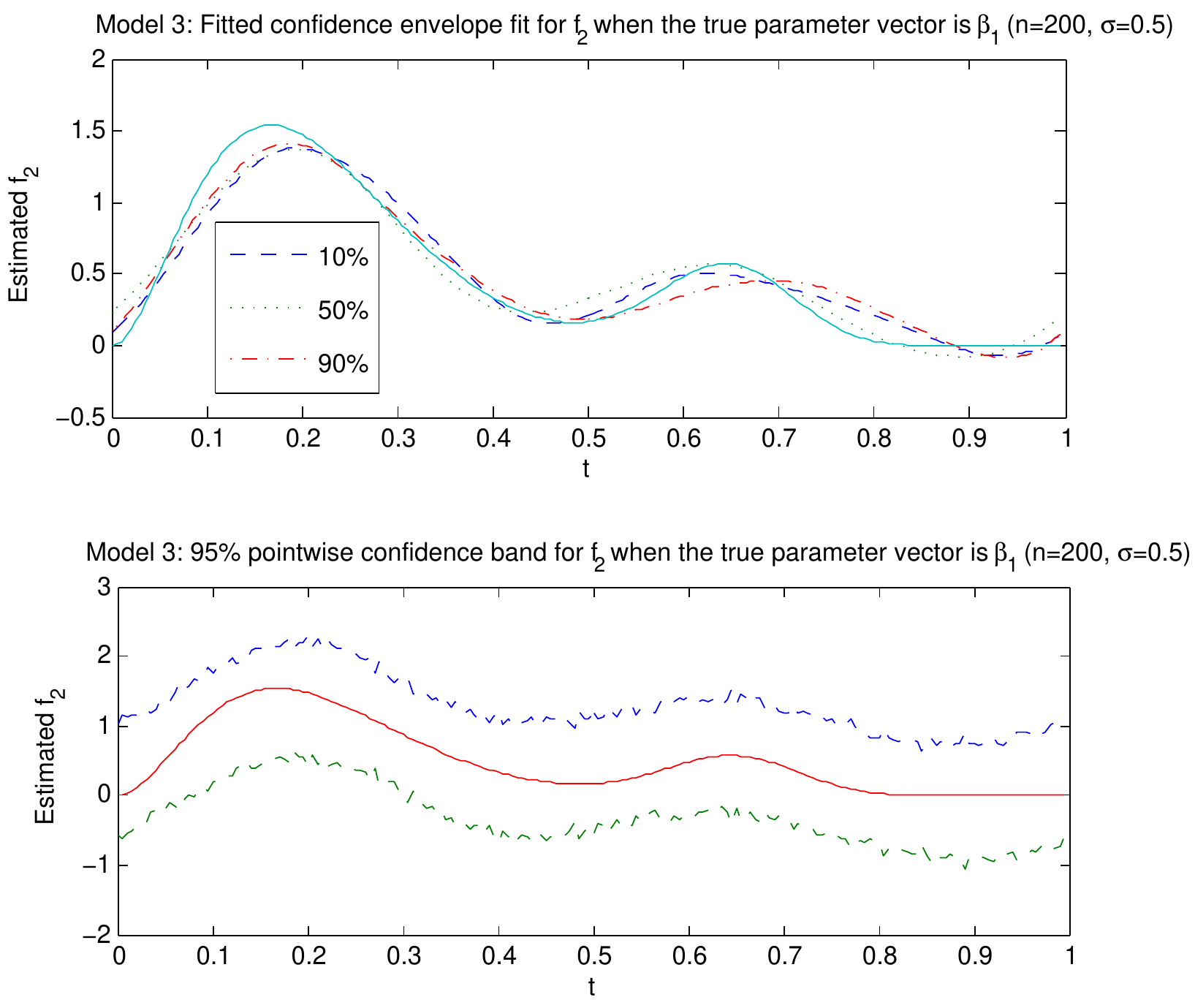}}
\caption{The estimated nonlinear functions given by the PSA in Model 3.}
\label{Ex3.plot2}
{\small The estimated nonlinear function,
confidence envelop and 95\% point-wise confidence interval for
Model 3 with true nonparametric function $f_2$ and true Euclidean parameter $\btv_1$,
$n=200$ and $\sigma=0.5$. In the top plot, the dashed
line is for the $10$th best fit, the dotted line is for the $50$th
best fit, and the dashed-dotted line is for the $90$th best among
500 simulations. The bottom plot is a $95\%$ pointwise confidence
interval.}
\end{figure}

\vspace{-0.8cm} \subsection{Real Example 1: Ragweed Pollen Data}\label{realData} \vskip
-0.6cm We apply the proposed method to the \textit{Ragweed Pollen}
data analyzed in \citet{r03}. The data consists of 87
daily observations of ragweed pollen level and relevant
information collected in Kalamazoo, Michigan during the 1993
ragweed season. The main purpose of this analysis is to develop an
accurate model for forecasting daily ragweed pollen level based on
some climate factors. The raw response $ragweed$ is the daily
ragweed pollen level (grains/$m^3$). There are four explanatory
variables: \vspace{-0.3cm}\begin{itemize} \item[] $X_1=$ rain: the
indicator of significant rain for the following day ($1=$ at least
3 hours of steady or brief but intense rain, $0=$ otherwise);
\item[] $X_2=$ temperature: temperature of the following day ($^o
F$); \item[] $X_3=$ wind: wind speed forecast for the following
day (knots); \item[] $X_4=$ day: the number of days in the current
ragweed pollen season.
\end{itemize}
\vspace{-0.3cm}We first standardize $X$-covariates. Since the raw
response is rather skewed, \citet{r03} suggested a
square root transformation $Y=\sqrt{ragweed}$. Marginal plots
suggest a strong nonlinear relationship between $Y$ and the {\it
day} number. Consequently, a partial linear model with a
nonparametric baseline $f(\text{day})$ is reasonable.
\citet{r03} fitted a semiparametric model with three linear effects
$X_1$, $X_2$ and $X_3$ and a nonlinear effect of $X_4$. For the
variable selection purpose, we add the quadratic terms in the
model and fit an enlarged model:
$$y=f(\mbox{day}) + \beta_1 x_1 + \beta_2 x_2 + \beta_3 x_3 + \beta_{4} x_2^2 + \beta_{4} x_3^2 + \varepsilon.$$
Table 9 gives the estimated regression coefficients. We observe that PSL and PSA end up
with the same model, and all the estimated coefficients are positive, suggesting that the ragweed pollen level increases as any
covariate increases. The shrinkage in parametric terms from the
partial spline models resulted from the PSA procedure is overall
smaller than that resulted from the PSL procedure.

\begin{table}[H]
\begin{center}
\caption{Estimated Coefficients for Ragweed Pollen Data}
\smallskip
\begin{tabular}{r|c|c|c}\hline
Covariate & PS & PSL & PSA \\\hline
rain & 1.3834 & 1.3620 & 1.3816\\
temperature & 0.1053  & 0.1045 & 0.1053 \\
wind & 0.2407 & 0.2384 & 0.2409\\
temp$\times$temp& 0.0042  & 0.0041 & 0.0041\\
wind$\times$wind &- 0.0004 & 0  & 0\\ \hline
\end{tabular}
\end{center}
\end{table}

Figure 6 depicts the estimated nonparametric function $\widehat{f}$(day)
and its 95\% pointwise confidence intervals given by the PSA. The
plot indicates that $\widehat{f}(\mathrm{day})$ increases rapidly to the
peak on around day 25, plunges until day 60, and decreases
steadily thereafter. The nonparametric fits given by the other two
procedures are similar and hence omitted in the paper.

We examined the prediction accuracy for PS, PSL and PSA,
in terms of the mean squared prediction errors (MSPE) based the leave-one-out strategy.
We also fit linear models for the above data using LASSO.
Our analysis shows that the MSPEs for PS, PSL and PSA are 5.63, 5.61 and 5.47 respectively,
while the MSPE for LASSO based on linear models is 12.40.
Roughly speaking, the MSPEs using PS, PSL and PSA are similar, though they provide
different model selection results summarized in Table 5. Notice that the PS method keeps more variables in the model than the PSL and PSA,
however the MSPEs are not much different. Thus, using PSL or PSA one can select a subgroup of significant variables to explain the model.
Furthermore,
the large MSPE based on linear models demonstrates invalidity of simply using linear models for such data.

\begin{figure}[H]
\centering{
\includegraphics[angle=0,width=0.6\textwidth,height=0.5\textwidth]{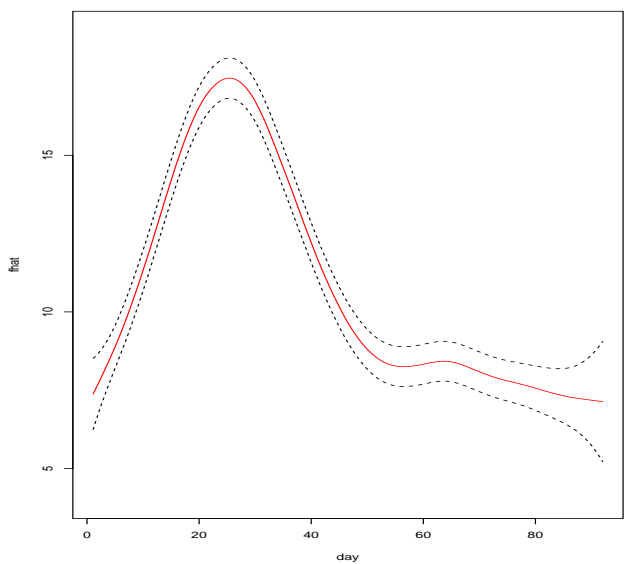}}
\caption{The estimated nonlinear function $\hat{f}(\text{day})$ for the ragweed pollen data.}
\label{real.plot}
{\small The estimated nonlinear function $\widehat{f}$(day) with
its $95\%$ pointwise confidence interval (dotted lines) given by the PSA for the ragweed pollen data.}
\end{figure}

\vspace{-0.8cm} \subsection{Real Example 2: Prostate Cancer Data}\label{realData} \vskip
-0.6cm We analyze the \textit{Prostate Cancer} data (\citet{s89}). The goal is
to predict the log level of prostate specific antigen using  a number of clinical measures.
The data consists of 97 men who were about to receive a radical prostatectomy.  There are eight predictors: $X_1=$ log cancer volume (lcavol),
$X_2=$ log prostate weight (lweight), $X_3=$ age, $X_4=$ log of benign prostatic hyperplasia amount (lbph),
$X_5=$ seminal vesicle invasion (svi), $X_6=$ log of capsular penetration (lcp), $X_7=$ Gleason score (gleason),
and $X_8=$ percent of Gleason scores of 4 or 5 (pgg45).

\begin{table}[H]
\begin{center}
\caption{Estimated Coefficients for Prostate Cancer Data}
\smallskip
\begin{tabular}{r|c|r|r}\hline
Covariate & PS & PSL & PSA \\\hline
lcavol & 0.587  & 0.443 & 0.562 \\
age & -0.020 & 0 & 0\\
lbph & 0.107  & 0 & 0\\
svi& 0.766 & 0.346  & 0.498\\
lcp & -0.105 & 0& 0\\
gleason & 0.045  & 0 & 0\\
pgg45& 0.005 & 0  & 0\\ \hline
\end{tabular}
\end{center}
\end{table}

A variable selection analysis was conducted in \citet{t96} using a linear regression model with LASSO,
and it selected three important variables {\sl lcavol, lweight, svi}
as important variables to predict the prostate specific antigen. We fitted partially linear models by treating {\sl lweight} as a
nonlinear term. Table 10 gives the estimated coefficients for different methods. Interestingly, both PSL and PSA select
{\sl lcavol} and {\sl svi} as important linear variables, which is consistent to the analysis by \citet{t96}.

\vspace{-0.8cm} \section{Discussion}\label{sec:disc} \vskip -0.6cm
We propose a new regularization method for simultaneous variable
selection for linear terms and component estimation for the nonlinear term in partial spline models. The
oracle properties of the new procedure for variable selection are
established. Moreover, we have shown that the new estimator can
achieve the optimal convergence rates for both the parametric and
nonparametric components. All the above conclusions are also
proven to hold in the increasing dimensional situation.

The proposed method sets up a basic framework to implement
variable selection for partial spline models, and it can be
generalized to other types of data analysis. In our future
research, we will generalize the results in this paper to the
generalized semiparametric models, robust linear regression, or
survival data analysis. In this paper, we assume the errors are
i.i.d. with constant variance, and the smoothness order of the Sobolev space is fixed as $m$,
though in practice we used $m=2$ to facilitate computation. In practice, the problem of
heteroscedastic error, i.e. the variance of $\epsilon$ is some
non-constant function of $(X,T)$, is often encountered. Meanwhile, the order $m$
may not be always available which needs to be approximated. We will
examine the latter two issues in the future.


\vspace{-0.6cm} \section{Proofs}\label{sec:proof} \vskip -0.6cm
For simplicity, we use $\widehat{\btv}$, $\widehat{\btv}_{1}$ $(\widehat{\btv}_{2})$ and $\hat{f}$
to represent $\widehat{\btv}_{PSA}$, $\widehat{\btv}_{PSA,1}$ $(\widehat{\btv}_{PSA,2})$ and $\hat{f}_{PSA}$,
in the proofs.

{\it Definition:} Let $\mathcal{A}$ be a subset of a (pseudo-)
metric space $(\mathcal{L},d)$ of real-valued functions. The
$\delta$-covering number $N(\delta,\mathcal{A},d)$ of
$\mathcal{A}$ is the smallest $N$ for which there exist functions
$a_{1},\ldots,a_{N}$ in $\mathcal{L}$, such that for each
$a\in\mathcal{A}$, $d(a,a_{j})\leq\delta$ for some
$j\in\{1,\ldots,N\}$. The $\delta$-bracketing number
$N_{B}(\delta,\mathcal{A},d)$ is the smallest $N$ for which there
exist pairs of functions
$\{[a_{j}^{L},a_{j}^{U}]\}_{j=1}^{N}\subset\mathcal{L}$, with
$d(a_{j}^{L},a_{j}^{U})\leq\delta$, $j=1,\ldots,N$, such that for
each $a\in\mathcal{A}$ there is a $j\in\{1,\ldots,N\}$ such that
$a_{j}^{L}\leq a\leq a_{j}^{U}$. The $\delta$-entropy number
($\delta$-bracketing entropy number) is defined as
$H(\delta,\mathcal{A},d)=\log N(\delta,\mathcal{A},d)$
($H_{B}(\delta,\mathcal{A},d)=\log N_{B}(\delta,\mathcal{A},d)$).

{\it Entropy Calculations:} For each $0<C<\infty$ and $\delta>0$,
we have \vspace{-.1in}\begin{eqnarray}
H_{B}(\delta,\{\eta:\|\eta\|_{\infty}\leq C, J_{\eta}\leq
C\},\|\cdot\|_{\infty})\leq M\left(\frac{C}{\delta}\right)^{1/m},\label{sob1}\\
H(\delta,\{\eta:\|\eta\|_{\infty}\leq C, J_{\eta}\leq
C\},\|\cdot\|_{\infty})\leq
M\left(\frac{C}{\delta}\right)^{1/m},\label{sob2}
\end{eqnarray}
where $\|\cdot\|_{\infty}$ represents the uniform norm and $M$ is
some positive number.

\vspace{-.1in}\begin{proof}[Proof of Theorem~\ref{brate}:]

In the proof of (\ref{betarate}), we will first show for any given
$\epsilon>0$, there exists a large constant $M$ such that
\begin{eqnarray}
P\left\{\inf_{\|\sv\|=M}\Delta(\sv)>0\right\}\geq
1-\epsilon,\label{inter00}
\end{eqnarray}
where $\Delta(\sv)\equiv Q(\btv_{0}+n^{-1/2}\sv)-Q(\btv_{0})$.
This implies with probability at least $(1-\epsilon)$ that there
exists a local minimum in the ball $\{\btv_{0}+n^{-1/2}\sv:
\|\sv\|\leq M\}$. Thus, we can conclude that there exists a local
minimizer such that $\|\hat{\btv}_{n}-\btv_{0}\|=O_{P}(n^{-1/2})$
if (\ref{inter00}) holds. Denote the quadratic part of $Q(\btv)$
as $L(\btv)$, i.e.,
$$L(\btv)=\frac{1}{n}(\yv-\Xv\btv)'[I-A(\lambda_1)](\yv-\Xv\btv).$$
Then we can obtain the below inequality:
\begin{eqnarray*}
\Delta(\sv)\geq
L(\btv_{0}+n^{-1/2}\sv)-L(\btv_{0})+\lambda_{2}\sum_{j=1}^{q}\frac{|\beta_{0j}+n^{-1/2}s_{j}|-|\beta_{0j}|}{|\tilde{\beta}_{j}|^{\gamma}},
\end{eqnarray*}
where $s_{j}$ is the $j$-th element of vector $\sv$. Note that
$L(\btv)$ is a quadratic function of $\btv$. Hence, by the Taylor
expansion of $L(\btv)$, we can show that
\begin{eqnarray}
\Delta(\sv)\geq
n^{-1/2}\sv'\dot{L}(\btv_{0})+\frac{1}{2}\sv'[n^{-1}\ddot{L}(\btv_{0})]\sv+
\lambda_{2}\sum_{j=1}^{q}\frac{|\beta_{0j}+
n^{-1/2}s_{j}|-|\beta_{0j}|}{|\tilde{\beta}_{j}|^
{\gamma}},\label{le3}
\end{eqnarray}
where $\dot{L}(\btv_{0})$ and $\ddot{L}(\btv_{0})$ are the first
and second derivative of $L(\btv)$ at $\btv_{0}$, respectively.
Based on (\ref{par3}), we know that
$-\dot{L}(\btv_{0})=(2/n)\Xv'[I-A(\lambda_{1})](\yv-\Xv\btv_{0})$
and $\ddot{L}(\btv_{0})=(2/n)\Xv'[I-A(\lambda_{1})]\Xv$. Combing
the proof of Theorem 1 and its four propositions in \citet{h86},
we can show that
\begin{eqnarray*}
n^{-1/2}\Xv'[I-A(\lambda_{1})](f_{0}+\epsilon)&&\overset{d}
{\longrightarrow}N(0,\sigma^{2}\bR),\\
n^{-1/2}\Xv'A(\lambda_{1})\epsilon&&\overset{P}{\longrightarrow}0.
\end{eqnarray*}
provided that $\lambda_{1}\rightarrow 0$ and
$n\lambda_{1}^{1/2m}\rightarrow\infty$. Therefore, the
Slutsky's theorem implies that
\vspace{-.1in}\begin{eqnarray}
\dot{L}(\btv_{0})&=&O_{P}(n^{-1/2}),\label{le1}\\
\ddot{L}(\btv_{0})&=&O_{P}(1)\label{le2}
\end{eqnarray}
given the above conditions on $\lambda_{1}$. Based on (\ref{le1})
and (\ref{le2}), we know the first two terms in the right hand
side of (\ref{le3}) are of the same order, i.e. $O_{P}(n^{-1})$.
And the second term, which converges to some positive constant,
dominates the first one by choosing sufficiently large $M$. The
third term is bounded by $n^{-1/2}\lambda_{2}M_{0}$ for some
positive constant $M_{0}$ since $\tilde{\beta}_{j}$ is the
consistent estimate for the nonzero coefficient for
$j=1,\ldots,q$. Considering that $\sqrt{n}\lambda_{2}\rightarrow
0$, we have completed the proof of (\ref{betarate}).

We next show the convergence rate for $\hat{f}$ in terms of
$\|\cdot\|_{n}$-norm, i.e. (\ref{fratee}). Let
$g_{0}(x,t)=\xv'\btv_{0}+f_{0}(t)$, and
$\hat{g}(x,t)=\xv'\hat{\btv}+\hat{f}(t)$. Then, by the definition
of $(\hat{\btv},\hat{f})$, we have \vspace{-.1in}\begin{eqnarray}
\|\hat{g}-g_{0}\|_{n}^{2}+\lambda_{1}J_{\hat{f}}^{2}+\lambda_{2}J_{\hat{\btv}}&\leq&
\frac{2}{n}\sum_{i=1}^{n}\epsilon_{i}(\hat{g}-g_{0})(X_{i},t_{i})+\lambda_{1}J_{f_{0}}^{2}+\lambda_{2}J_{\btv_{0}},
\label{ineq10}\\
\|\hat{g}-g_{0}\|_{n}^{2}&\leq&2\|\epsilon\|_{n}\|
\hat{g}-g_{0}\|_{n}+\lambda_{1}J_{f_{0}}^{2}+\lambda_{2}J_{\btv_{0}},\nonumber\\
\|\hat{g}-g_{0}\|_{n}^{2}&\leq&\|\hat{g}-g_{0}\|_{n}O_{P}(1)+o_{P}(1),\nonumber
\end{eqnarray}
where
$J_{\beta}\equiv\sum_{j=1}^{d}|\beta_{j}|/|\tilde{\beta}_{j}|^{\gamma}$.
The second inequality follows from the Cauchy-Schwartz inequality.
The last inequality holds since $\epsilon$ has sub-exponential
tail, and $\lambda_{1}, \lambda_{2}\rightarrow 0$. Then the above
inequality implies that $\|\hat{g}-g_{0}\|_{n}=O_{P}(1)$, so that
$\|\hat{g}\|_{n}=O_{P}(1)$. By Sobolev embedding theorem, we can
decompose $g(x,t)$ as $g_{1}(x,t)+g_{2}(x,t)$, where
$g_{1}(x,t)=x'\btv+\sum_{j=1}^{m}\alpha_{j}t^{j-1}$ and
$g_{2}(x,t)=f_{2}(t)$ with $\|g_{2}(x,t)\|_{\infty}\leq
J_{g_{2}}=J_{f}$. Similarly, we can write
$\hat{g}=\hat{g}_{1}+\hat{g}_{2}$, where
$\hat{g}_{1}=x'\hat{\btv}+\sum_{j=1}^{m}\hat{\alpha}_{j}t^{j-1}=\hat{\delta}'\phi$
and $\|\hat{g}_{2}\|_{\infty}\leq J_{\hat{g}}$. We shall now show
that $\|\hat{g}\|_{\infty}/(1+J_{\hat{g}})=O_{P}(1)$ via the above
Sobolev decomposition. Then
\begin{eqnarray}
\frac{\|\hat{g}_{1}\|_{n}}{1+J_{\hat{g}}}\leq\frac{\|\hat{g}\|_{n}}{1+J_{\hat{g}}}+\frac{\|\hat{g}_{2}\|_{n}}{1+J_{\hat{g}}}=O_{P}(1).\label{inter1}
\end{eqnarray}
Based on the assumption about $\sum_{k}\phi_{k}\phi_{k}'/n$,
(\ref{inter1}) implies that
$\|\hat{\delta}\|/(1+J_{\hat{g}})=O_{P}(1)$. Since $(X,t)$ is in a
bounded set, $\|\hat{g}_{1}\|_{\infty}/(1+J_{\hat{g}})=O_{P}(1)$.
So we have proved that
$\|\hat{g}\|_{\infty}/(1+J_{\hat{g}})=O_{P}(1)$. Thus, the
entropy calculation (\ref{sob1}) implies that
\vspace{-.1in}\begin{eqnarray*}
H_{B}\left(\delta,\left\{\frac{g-g_{0}}{1+J_{g}}:g\in\mathcal{G},\frac{\|g\|_{\infty}}{1+J_{g}}\leq
C\right\},\|\cdot\|_{\infty}\right)\leq M_{1}\delta^{-1/m},
\end{eqnarray*}
where $M_{1}$ is some positive constant, and
$\mathcal{G}=\{g(x,t)=x'\btv+f(t):\btv\in R^d, J_{f}<\infty\}$.
Based on Theorem 2.2 in \citet{mvg97} about the continuity modulus
of the empirical processes
$\{\sum_{i=1}^{n}\epsilon_{i}(g-g_{0})(z_{i})\}$ indexed by $g$
and (\ref{ineq10}), we can establish the following set of
inequalities: \vspace{-.05in}\begin{eqnarray}
\lambda_{1}J_{\hat{f}}^{2}&\leq&
\left[\|\hat{g}-g_{0}\|_{n}^{1-1/2m}(1+J_{\hat{f}})^{1/2m}\vee(1+J_{\hat{f}})n^{-\frac{2m-1}{2(2m+1)}}\right]O_{P}(n^{-1/2})\nonumber
\\&&+\lambda_{1}J_{f_{0}}^{2}+\lambda_{2}(J_{\btv_{0}}-J_{\hat{\btv}}),\label{inter2}
\end{eqnarray}
and
\begin{eqnarray}
\|\hat{g}-g_{0}\|_{n}^{2}&\leq&
\left[\|\hat{g}-g_{0}\|_{n}^{1-1/2m}(1+J_{\hat{f}})^{1/2m}\vee(1+J_{\hat{f}})n^{-\frac{2m-1}{2(2m+1)}}\right]O_{P}(n^{-1/2})\label{inter3}
\\&&+\lambda_{1}J_{f_{0}}^{2}+\lambda_{2}(J_{\btv_{0}}-J_{\hat{\btv}}).\nonumber
\end{eqnarray}
Note that \vspace{-.1in}\begin{eqnarray}
\lambda_{2}(J_{\btv_{0}}-J_{\hat{\btv}})&\leq&\lambda_{2}\sum_{j=1}^{q}\frac{|\beta_{0j}-\hat{\beta}_{j}|}{|\tilde{\beta}_{j}|^{\gamma}}+
\lambda_{2}\sum_{j=q+1}^{d}\frac{|\beta_{0j}-\hat{\beta}_{j}|}{|\tilde{\beta}_{j}|^{\gamma}}\nonumber\\
&\leq& O_{P}(n^{-2m/(2m+1)}).\label{inter4}
\end{eqnarray}
(\ref{inter4}) in the above follows from
$\|\widehat{\btv}-\btv_{0}\|=O_{P}(n^{-1/2})$ and (\ref{smooth}).
Thus, solving the above two inequalities gives
$\|\hat{g}-g_{0}\|_{n}=O_{P}(\lambda_{1}^{1/2})$ and
$J_{\hat{f}}=O_{P}(1)$ when
$n^{2m/(2m+1)}\lambda_{1}\rightarrow\lambda_{10}>0$. Note that
$$\|X'(\hat{\btv}-\btv_{0})\|_{n}=\sqrt{(\hat{\btv}-\btv_{0})'(\sum_{i=1}^{n}X_{i}X'_{i}/n)(\hat{\btv}-\btv_{0})}\aplt\|\hat{\btv}
-\btv_{0}\|=O_{P}(n^{-1/2})$$ by (\ref{betarate}). Applying the
triangle inequality to
$\|\hat{g}-g_{0}\|_{n}=O_{P}(\lambda_{1}^{1/2})$, we have proved
that $\|\hat{f}-f_{0}\|_{n}=O_{P}(\lambda_{1}^{1/2})$.

We next prove 3(a). It suffices to show that
\begin{eqnarray}
Q\{(\bar{\btv}_{1},\zerov)\}=\min_{\|\bar{\btv}_{2}\|\leq
Cn^{-1/2}}Q\{(\bar{\btv}_{1},\bar{\btv}_{2})\}\;\;\mbox{with
probability approaching to 1}\label{inter5}
\end{eqnarray}
for any $\bar{\btv}_{1}$ satisfying
$\|\bar{\btv}_{1}-\btv_{1}\|=O_{P}(n^{-1/2})$ based on
(\ref{betarate}). To show (\ref{inter5}), we need to show
that $\partial Q(\btv)/\partial\beta_{j}<0$ for
$\beta_{j}\in(-Cn^{-1/2},0)$, and $\partial
Q(\btv)/\partial\beta_{j}>0$ for $\beta_{j}\in(0,Cn^{-1/2})$, for
$j=q+1,\ldots,d$, holds with probability tending to 1. By two term
Taylor expansion of $L(\btv)$ at $\btv_{0}$, $\partial
Q(\btv)/\partial\beta_{j}$ can be expressed in the following form
for $j=q+1,\ldots,d$:
\begin{eqnarray*}
\frac{\partial Q(\btv)}{\partial\beta_{j}}=\frac{\partial
L(\btv_{0})}{\partial\beta_{j}}+\sum_{k=1}^{d}\frac{\partial^{2}L(\btv_{0})}{\partial\beta_{j}\partial\beta_{k}}(\beta_{k}-\beta_{0k})
+\lambda_{2}\frac{1\times
sgn(\beta_{j})}{|\tilde{\beta}_{j}|^{\gamma}},
\end{eqnarray*}
where $\beta_{k}$ is the $k^{\text{th}}$ element of vector $\btv$.
Note that $\|\btv-\btv_{0}\|=O_{P}(n^{-1/2})$ by the above
constructions. Hence , we have
\begin{eqnarray*}
\frac{\partial
Q(\btv)}{\partial\beta_{j}}=O_{P}(n^{-1/2})+sgn(\beta_{j})\frac{\lambda_{2}}{|\tilde{\beta}_{j}|^{\gamma}}
\end{eqnarray*}
by (\ref{le1}) and (\ref{le2}) in the above. The assumption (\ref{smooth}) implies that
$\sqrt{n}\lambda_{2}/|\tilde{\beta}_{j}|^{\gamma}\rightarrow\infty$
for $j=q+1,\ldots,d$. Thus, the sign of $\beta_{j}$
determines that of $\partial Q(\btv)/\partial\beta_{j}$ for
$j=q+1,\ldots,d$. This completes the proof of 3(a).

Now we prove 3(b). Following similar proof of (\ref{betarate}), we can show that there exists a
$\sqrt{n}$ consistent local minimizer of $Q(\btv_{1},0)$, i.e.
$\hat{\btv}_{1}$, and satisfies:
\begin{eqnarray*}
\frac{\partial Q(\btv)}{\partial\beta_{j}}|_{\btv=(\hat{\btv}_{1},
\zerov)}=0
\end{eqnarray*}
for $j=1,\ldots,q$. By similar analysis in the above, we can
establish the equation:
\begin{eqnarray*}
0=\frac{\partial
L(\btv_{0})}{\partial\beta_{j}}+\sum_{k=1}^{q}\left\{\frac{\partial^{2}L(\btv_{0})}{\partial\beta_{j}\partial\beta_{k}}\right\}(\hat{\beta}_{k}-
\beta_{0k})+\lambda_{2}\frac{1\times
sgn(\hat{\beta}_{j})}{|\tilde{\beta}_{j}|^{\gamma}},
\end{eqnarray*}
for $j=1,\ldots,q$. Note that the assumption
$\sqrt{n}\lambda_{2}\rightarrow 0$ implies that the third term in
the right hand side of the above equation is $o_{P}(n^{-1/2})$. By
the form of $L(\btv)$ and the Slutsky's
theorem, we conclude the proof of 3(b).
\end{proof}

\vspace{-.2in}\begin{proof}[Important Lemmas]


We provide three useful matrix inequalities and two lemmas for
preparing the proofs of Theorems~\ref{consithm} and
\ref{asynorthm}. Given any $n\times m$ matrix $\Av$ and symmetric
strictly positive definite matrix $\Bv$, $n\times 1$ vector $\sv$
and $\zv$, and $m\times 1$ vector $\wv$, we have
\begin{eqnarray}
|\sv'\Av\wv|&\leq&\|\sv\|\|\Av\|\|\wv\|\label{ineq1}\\
|\sv'\Bv\zv|&\leq&|\sv'\Bv\sv|^{1/2}|\zv'\Bv\zv|^{1/2}\label{ineq3}\\
|\sv'\zv|&\leq&\|\sv\|\|\zv\|\label{ineq2}
\end{eqnarray}
where $\|\Av\|^2=\sum_{j}\sum_{i}a_{ij}^{2}$. (\ref{ineq3})
follows from the Cauchy-Schwartz inequality.

\begin{lemma}
Given that $\lambda_{1}\rightarrow 0$, we have
\begin{eqnarray}
n^{-k/2}\sum_{l=1}^{n}|[(I-A)\fv_{0}(t)]_{l}|^{k}=O(\lambda_{1}^{k/2})\;\;\;\;\;\;\mbox{for}\;\;k=2,3,\ldots\label{i-a}
\end{eqnarray}
\end{lemma}

{\bf Proof:} For the case of $k=2$, it has been proved in Lemma 2
of \citet{h86}. Next we apply the principle of mathematical
induction to prove the cases for arbitrary $k>2$. We first assume
that
\begin{eqnarray}
n^{-(k-1)/2}\sum_{l=1}^{n}|[(I-A)\fv_{0}(t)]_{l}|^{k-1}=O(\lambda_{1}^{(k-1)/2})\label{mathind}
\end{eqnarray}
for $k=3$. Then we can write
\begin{eqnarray*}
&&n^{-k/2}\sum_{l=1}^{n}|[(I-A)\fv_{0}(t)]_{l}|^{k}\\
\leq&& n^{-1/2}\max_{l=1,\ldots,n}|[(I-A)\fv_{0}(t)]_{l}|\times
n^{-(k-1)/2}\sum_{l=1}^{n}|[(I-A)\fv_{0}(t)]_{l}|^{k-1}\\
\leq&&n^{-1/2}\left[\sum_{l=1}^{n}[(I-A)\fv_{0}(t)]_{l}^{2}\right]^{1/2}\times
O(\lambda_{1}^{(k-1)/2})=O(\lambda_{1}^{k/2}).
\end{eqnarray*}
The last step follows from (\ref{mathind}) and the case for $k=2$.
$\Box$

\begin{lemma}\label{apple1}
Given that $d_{n}\le n^{1/2}\wedge n\lambda_{1}^{1/2m}$, we have
\begin{eqnarray}
\left[\Xv'A(\lambda_{1})\epsv\right]_{i}&=&
O_{P}(\lambda_{1}^{-1/4m}),\label{inter1a}\\
\left[\Xv'((I-A(\lambda_{1}))\fv_{0}+\epsv)\right]_{i}&=&O_{P}(n^{1/2}),\label{inter4a}\\
\left[\Xv'(I-A(\lambda_{1}))\Xv/n\right]_{ij}&=&\bR_{ij}+O_{P}(n^{-1/2}\vee
n^{-1}\lambda_{1}^{-1/2m}),\label{inter2a}\\
\|\Xv'(I-A(\lambda_{1}))\Xv/n-R\|&=&o_{P}(1).\label{inter3a}
\end{eqnarray}
\end{lemma}
{\bf Proof:} We first state the Lemma 4.1 and 4.3 in \citet{cw79}:
\begin{eqnarray}
n^{-1}\sum_{j}[(I-A)\fv_{0}]_{j}^{2}\leq\lambda_{1}\int_{0}^{1}(f_{0}^{(m)}(t))^{2}dt,\label{cwres0}\\
tr(A)=O(\lambda_{1}^{-1/2m})\;\;\;\mbox{and}\;\;\;tr(A^{2})=O(\lambda_{1}^{-1/2m}).\label{cwres}
\end{eqnarray}
By the fact that
$Var[(\Xv'A\epsv)_{i}]=\sigma^{2}R_{ii}tr(A^{2})$, we can show
that $[\Xv'A\epsv]_{i}=O_{P}(\lambda_{1}^{-1/4m})$ based on
(\ref{cwres}), thus proved (\ref{inter1a}). We first write the
left hand side of (\ref{inter4a}) as
$\sqrt{n}\sum_{j=1}^{n}W_{ij}$, where
\begin{eqnarray*}
W_{ij}=n^{-1/2}X_{ij}(\epsilon_{j}+((I-A)\fv_{0})_{j})\;\;
\mbox{and}\;\;X_{ij}\;\mbox{is the}\;(j,i)-th\;\mbox{element
of}\;\Xv
\end{eqnarray*}
for $i=1,\ldots,d_{n}$. We next apply the Lindeberg's theorem to
$\sum_{j}W_{ij}$. It is easy to show that
$Var(\sum_{j}W_{ij})=\bR_{ii}\sigma^{2}+\bR_{ii}n^{-1}\sum_{j}[(I-A)\fv_{0}]_{j}^{2}$.
By (\ref{cwres0}), we have $Var(\sum_{j}W_{ij})\rightarrow
\bR_{ii}\sigma^{2}$. We next verify the  Liapounov's condition:
\begin{eqnarray*}
\sum_{j}E|W_{ij}|^{3}&=&n^{-3/2}E|X_{ij}|^{3}\sum_{j}E|\epsilon_{j}+[(I-A)\fv_{0}]_{j}|^{3}\\
&\leq&3n^{-3/2}\left[nE|\epsilon|^{3}+\sum_{j}|[(I-A)\fv_{0}]_{j}|^{3}\right]\rightarrow
0
\end{eqnarray*}
by the sub-exponential tail of $\epsilon$ and (\ref{i-a}). Then
the Lindeberg's theorem implies (\ref{inter4a}). As for
(\ref{inter2a}), we first write (\ref{inter2a}) as the sum of
$\bR_{ij}$, $[\Xv'\Xv/n]_{ij}-\bR_{ij}$ and $[-\Xv'A\Xv/n]_{ij}$.
By the central limit theorem, the second term in the above
decomposition is $O_{P}(n^{-1/2})$. For the last term, we have
$E\{(\Xv' A\Xv)_{ij}\}^{2}=$
\begin{eqnarray*}
(\bR_{ij})^{2}(tr(A))^{2}+(\bR_{ii}\bR_{jj}+(\bR_{ij})^{2})tr(A^{2})+
(E(X_{1i}X_{1j})^{2}-2(\bR_{ij})^{2}-\bR_{ii}\bR_{jj})\sum_{r}A_{rr}^{2}
\end{eqnarray*}
for $i\neq j$. When $i=j$, we have $E|(\Xv'
A\Xv)_{ii}|=\bR_{ii}tr(A)$. By considering (\ref{cwres}) we have
proved (\ref{inter2a}). (\ref{inter2a}) implies that
\begin{eqnarray}
\|\Xv'(I-A)\Xv/n-\bR\|=O_{P}(d_{n}n^{-1/2}\vee
d_{n}n^{-1}\lambda_{1}^{-1/2m}).\label{cova}
\end{eqnarray}
Thus (\ref{inter3a}) follows from the dimension condition D1.
\end{proof}


\vspace{-.2in}\begin{proof}[Proof of Lemma~\ref{intilemma}:]

Based on the definition on $\tilde{\btv}_{PS}$, we have the below
inequality:
\begin{eqnarray*}
\frac{1}{n}(\tilde{\btv}_{PS}-\btv_{0})'\Xv'(I-A)\Xv(\tilde{\btv}_{PS}-\btv_{0})-\frac{2}{n}(\tilde{\btv}_{PS}-\btv_{0})'\Xv'(I-A)
(\fv_{0}+\epsv)\leq 0.
\end{eqnarray*}
Let
$\delta_{n}=n^{-1/2}[\Xv'(I-A)\Xv]^{1/2}(\tilde{\btv}_{PS}-\btv_{0})$
and
$\omega_{n}=n^{-1/2}[\Xv'(I-A)\Xv]^{-1/2}\Xv'(I-A)(\fv_{0}+\epsv)$.
Then the above inequality can be rewritten as $
\|\delta_{n}\|^{2}-2\omega_{n}'\delta_{n}\leq 0$, i.e.
$\|\delta_{n}-\omega_{n}\|^{2}\leq\|\omega_{n}\|^{2}$. By
Cauchy-Schwartz inequality, we have $\|\delta_{n}\|^{2}\leq
2(\|\delta_{n}-\omega_{n}\|^{2}+\|\omega_{n}\|^{2})\leq4\|\omega_{n}\|^{2}$.
Examine $\|\omega_{n}\|^{2}=K_{1n}+K_{2n}+K_{3n}$, with
\begin{eqnarray*}
K_{1n}&=&n^{-1}\epsv'(I-A)\Xv[\Xv'(I-A)\Xv]^{-1}
\Xv'(I-A)\epsv\\
K_{2n}&=&2n^{-1}\epsv'(I-A)\Xv[\Xv'(I-A)\Xv]^{-1}
\Xv'(I-A)\fv_{0}(t)\\
K_{3n}&=&n^{-1}\fv_{0}(T)'(I-A)\Xv[\Xv'(I-A)\Xv]^{-1}\Xv'(I-A)
\fv_{0}(t).
\end{eqnarray*}
Applying (\ref{inter1a}), (\ref{inter4a}) and (\ref{inter2a}) to
the above three terms, we can conclude that all of them are of the
order $O_{P}(d_{n}n^{-1})$ by considering the matrix inequalities
(\ref{ineq1})-(\ref{ineq2}). Thus we have proved (\ref{inirate})
by considering (\ref{inter3a}). \end{proof}


\vspace{-.2in}\begin{proof}[Proof of Theorem~\ref{consithm}:]
The proof proceeds in several parts. First we show the rate convergence of the PSA parametric estimate,
i.e., (\ref{pararatei}). Second, we derive the rate of convergence for $\widehat{f}$.

Let
$\alpha_{n}=\sqrt{d_{n}/n}$. Similar as (\ref{le3}), we have
\begin{eqnarray}
Q(\btv_0+\alpha_n\sv)-Q(\btv_0)\geq
\alpha_{n}\sv'\dot{L}(\btv_{0})+\frac{1}{2}\sv'[\alpha_{n}^{2}\ddot{L}(\btv_{0})]\sv+
\lambda_{2}\sum_{j=1}^{q_{n}}\frac{|\beta_{0j}+\alpha_n
s_{j}|-|\beta_{0j}|}{|\tilde{\beta}_{j}|^{\gamma}},\label{old}
\end{eqnarray}
where the forms of $\dot{L}(\btv_{0})$ and $\ddot{L}(\btv_{0})$
are specified in the proof of Theorem~\ref{brate}. By considering
the lemma \ref{apple1}, (\ref{ineq1}) and (\ref{ineq2}) in the
appendix, we have
\begin{eqnarray}
\alpha_{n}\sv'\dot{L}(\btv_{0})&=&\|\sv\|O_{P}(d_{n}/n)\\
\frac{1}{2}\sv'[\alpha_{n}^{2}\ddot{L}(\btv_{0})]\sv&=&
(d_{n}/n)\sv'\bR\sv+O_P(d_{n}^{2}n^{-3/2}\vee
d_{n}^{2}n^{-2}\lambda_{1}^{-1/2m})
\end{eqnarray}
given any $\|\sv\|=C$ independent of $n$. Thus the first two terms
in the right hand side of (\ref{old}) are of the same order
$O_{P}(d_{n}/n)$ due to $d_n=o(n^{1/2}\wedge
n\lambda_1^{1/2m})$. The second term, which is positive, dominates
the first one by allowing sufficiently large $C$. The
last term is bounded by $\lambda_{2}\alpha_n\|\sv\|$. Thus, we
assume $\sqrt{n}\lambda_{2}/\sqrt{d_{n}}\rightarrow 0$ so that the
last term of (\ref{old}) is $o_{P}(d_{n}/n)$. This completes the proof of (\ref{pararatei}).

We next show the nonparametric rate for $\widehat{f}$ by using
similar analysis for the fixed dimensional case. Recall that
$g(x,t)=x'\btv+f(t)$. Similarly, we can show
$\|\widehat{g}-g_{0}\|_{n}=O_{P}(1)$. Combining the fact that
$\|g_{0}\|_{\infty}=O_{P}(q_{n})$, we have
$\|\widehat{g}\|_{n}=O_{P}(q_{n})$. By assuming that
$\lambda_{min}(\sum_{k}\phi_{k}\phi_{k}'/n)\geq c_{3}>0$, we can
obtain
\begin{eqnarray*}
\frac{\|\hat{g}\|_{\infty}}{1+J_{\hat{g}}}=O_{P}\left(\frac{q_{n}}{1+J_{\hat{g}}}\right)
\end{eqnarray*}
by similar analysis. Thus, by applying Theorem 2.2 in
\citet{mvg97}, we have established the below inequalities:
\begin{eqnarray}
\lambda_{1}J_{\hat{f}}^{2}&\leq&
\left[\|\hat{g}-g_{0}\|_{n}^{1-1/2m}(1+J_{\hat{f}})^{1/2m}q_{n}^
{1/2m}\vee(1+J_{\hat{f}})
q_{n}n^{-\frac{2m-1}{2(2m+1)}}\right]O_{P}(n^{-1/2})\nonumber
\\&&+\lambda_{1}J_{f_{0}}^{2}+\lambda_{2}(J_{\btv_{0}}-
J_{\widehat{\btv}}),\label{ratein0}
\end{eqnarray}
\begin{eqnarray}
\|\hat{g}-g_{0}\|_{n}^{2}&\leq&
\left[\|\hat{g}-g_{0}\|_{n}^{1-1/2m}(1+J_{\hat{f}})^{1/2m}
q_{n}^{1/2m}\vee(1+J_{\hat{f}})q_{n}
n^{-\frac{2m-1}{2(2m+1)}}\right]O_{P}(n^{-1/2})
\nonumber\\&&+\lambda_{1}J_{f_{0}}^{2}+\lambda_{2}
(J_{\btv_{0}}-J_{\widehat{\btv}}).\label{ratein2}
\end{eqnarray}
Let $a_{n}=\|\widehat{g}-g_{0}\|_{n}/[(1+J_{\hat{f}})q_{n}]$, then from $(1+J_{\hat{f}})q_n\ge1$,
(\ref{ratein2}) becomes
\begin{eqnarray}
a_{n}^{2}&\le& a_n^2(1+J_{\hat{f}})q_n\nonumber\\
&\leq& O_{P}(n^{-1/2})a_{n}^{1-1/2m}\vee
O_{P}(n^{-2m/(2m+1)})\vee
O_{P}(\lambda_{1}/q_{n})\vee\frac{\lambda_{2}(J_{\btv_{0}}
-J_{\widehat{\btv}})}{q_{n}}
\nonumber\\
&\leq&O_{P}(n^{-1/2})a_{n}^{1-1/2m}\vee O_{P}(n^{-2m/(2m+1)})
\vee\frac{\lambda_{2}(J_{\btv_{0}}-J_{\widehat{\btv}})}
{q_{n}}\nonumber\\
&\leq&O_{P}(n^{-1/2})a_{n}^{1-1/2m}\vee
O_{P}(n^{-2m/(2m+1)}).\label{rateine1}
\end{eqnarray}
In view of the condition $\lambda_{1}/q_{n}\asymp n^{-2m/(2m+1)}$,
the second inequality in the above follows. The last inequality
follows from the below analysis. Note that
\begin{eqnarray*}
\frac{\lambda_{2}(J_{\btv_{0}}-J_{\widehat{\btv}})}{q_{n}}&\leq&
\left(\lambda_{2}\sum_{j=1}^{q_{n}}\frac{|\beta_{0j}-\hat{\beta}_{j}|}
{|\tilde{\beta}_{j}|^{\gamma}}+
\lambda_{2}\sum_{j=q_{n}+1}^{d_{n}}\frac{|\beta_{0j}-\hat{\beta}_{j}|}{|\tilde{\beta}_{j}|^{\gamma}}\right)q_{n}^{-1}
\nonumber\\
&\aplt&\left(\lambda_{2}\sum_{j=1}^{q_n}|\beta_{0j}-
\widehat{\beta}_j|+\max_{j=q_{n}+1,\ldots,d_{n}}\frac{\lambda_{2}}
{|\tilde{\beta}_{j}|^{\gamma}}
\sum_{j=q_n+1}^{d_n}|\beta_{0j}-\widehat{\beta}_j|
\right)q_{n}^{-1}\\
&\aplt&\left[\max_{j=q_{n}+1,\ldots,d_{n}}\frac{\lambda_{2}/q_n}
{|\tilde{\beta}_{j}|^{\gamma}}\right]O_{P}(\sqrt{d_n/n})\sqrt{d_n}\\
&=&O_P(n^{1/(2m+1)}d_n^{-3/2}\sqrt{d_n/n})\cdot O_{P}(\sqrt{d_n/n})\sqrt{d_n}\\
&=&O_P(n^{-2m/(2m+1)})
\end{eqnarray*}
since $\|\widehat{\btv}-\btv_{0}\|=O_{P}(\sqrt{d_{n}/n})$ and
(\ref{l2ratei2}). Therefore (\ref{rateine1}) implies that
$a_{n}=O_{P}(n^{-m/(2m+1)})$. We next analyze (\ref{ratein0})
which can be rewritten as
\begin{eqnarray*}
\frac{\lambda_{1}}{q_{n}}(J_{\hat{f}}-1)&\leq&
O_{P}(n^{-1/2})a_{n}^{1-1/2m}\vee
O_{P}(n^{-2m/(2m+1)})\\
(J_{\hat{f}}-1)&\leq&\frac{q_{n}}{\lambda_{1}}O_{P}(n^{-2m/(2m+1)})\\
J_{\hat{f}}&\leq&O_{P}(1).
\end{eqnarray*}
in view of the condition that $\lambda_{1}/q_{n}\asymp
n^{2m/(2m+1)}$. Finally, we have proved that
$\|\hat{g}-g_{0}\|_{n}=O_{P}(n^{-m/(2m+1)}q_{n})$. Combining the
triangle inequality and
$\|\widehat{\btv}-\btv_{0}\|=O_{P}(\sqrt{d_{n}/n})$, we complete
the whole proof of (\ref{nonpararatei}). \end{proof}

\vspace{-.2in}\begin{proof}[Proof of Theorem~\ref{asynorthm}:]


Proof of part (a) is similar as that in the fixed dimension
case, i.e. 3(a) in Theorem~\ref{brate}. It follows from the
regular condition $\lambda_{1}/q_{n}\asymp n^{-2m/(2m+1)}$,
Lemma~\ref{apple1} and assumption~(\ref{l2ratei1}).

We next prove the asymptotic normality of $\widehat{\btv}_{1}$.
Similar as the proof for 3(b) in Theorem~\ref{brate}, we
can establish that
\begin{eqnarray}
\widehat{\btv}_{1}-\btv_{10}=\left[\Xv_{1}'(I-A)\Xv_{1}\right]^{-1}\left[\Xv_{1}'(I-A)(\fv_{0}(t)+\epsv)-
\frac{n\lambda_{2}}{2}Pe(\widehat{\btv}_{1})\right],\label{disequ}
\end{eqnarray}
where $Pe(\widehat{\btv}_{1})=(sign(\widehat{\beta}_{1})/
|\tilde{\beta}_{1}|^{\gamma},\ldots,sign(\widehat{\beta}_{q_n})/
|\tilde{\beta}_{q_n}|^{\gamma})'$. Note that the invertibility of
$\Xv_{1}(I-A)\Xv_{1}$ follows from (\ref{inter3a}) and the
asymptotic invertibility of $\bR$, i.e. the condition R3D. Thus,
we have
\begin{eqnarray}
&&\sqrt{n}\Gv_{n}\bR_{11}^{-1/2}(\Xv_{1}'(I-A)\Xv_{1}/n)(\widehat{\btv}_{1}-\btv_{10})\label{interdis}\\
=&&\sqrt{n}\Gv_{n}\bR_{11}^{-1/2}\left[\frac{\Xv_{1}'(I-A)(\fv_{0}(t)+\epsv)}{n}-
\frac{\lambda_{2}}{2}Pe(\widehat{\btv}_{1})\right]\nonumber\\
=&&M_{1n}+M_{2n}+M_{3n},\nonumber
\end{eqnarray}
where
\begin{eqnarray*}
M_{1n}&=&n^{-1/2}\Gv_{n}\bR_{11}^{-1/2}\Xv_{1}'[(I-A)\fv_{0}(t)+
\epsv],\\
M_{2n}&=&-n^{-1/2}\Gv_{n}\bR_{11}^{-1/2}\Xv_{1}'A\epsv,\\
M_{3n}&=&-(\sqrt{n}\lambda_{2}/2)\Gv_{n}\bR_{11}^{-1/2}Pe(
\widehat{\btv}_{1}).
\end{eqnarray*}

In order to derive the asymptotic distribution of
$M_{1n}+M_{2n}+M_{3n}$, we apply the Cramer-Wold device. Let $\vv$
be a $l$-vector. We first show that $\vv'M_{2n}=o_{P}(1)$ and
$\vv'M_{3n}=o_{P}(1)$. It is easy to show
\begin{eqnarray*}
|\vv'M_{2n}|&&\leq
n^{-1/2}\|\vv\|\|\Gv_{n}\bR_{11}^{-1/2}\Xv_{1}'A\epsv\|\leq
(n\lambda_{min}(\bR_{11}))^{-1/2}\|\vv\|\|\Gv_{n}\Xv_{1}'A\epsv\|\\
&&\leq O_{P}(n^{-1/2}\sqrt{q_{n}}\lambda_{1}^{-1/4m})=o_{P}(1).
\end{eqnarray*}
The last inequality follows from $\Gv_{n}\Gv_{n}'\rightarrow\Gv$
and (\ref{inter1a}). The conditions that $\lambda_{1}/q_{n}\asymp
n^{-2m/(2m+1)}$ and $n^{m/(2m+1)}\lambda_{1}\rightarrow 0$ imply
its convergence to zero. As for $\vv'M_{3n}$, we have
\begin{eqnarray*}
|\vv'M_{3n}|&&\leq\frac{\sqrt{n}\lambda_{2}}{2}\|\vv\|\|\Gv_{n}
\bR_{11}^{-1/2}Pe(\widehat{\btv}_{1})\|\leq
O_{P}(\sqrt{n}\lambda_{2})\|\Gv_{n}Pe(\widehat{\btv}_{1})\|\leq
O_{P}(\sqrt{n}\lambda_{2}\sqrt{q_{n}})=o_{P}(1)
\end{eqnarray*}
by the stated condition $q_{n}=o(n^{-1}\lambda_{2}^{-2})$.

As for $\vv'M_{1n}$, we can rewrite it as
\begin{eqnarray*}
\vv'M_{1n}=\sum_{j=1}^{n}n^{-1/2}\vv'\Gv_{n}\bR_{11}^{-1/2}\wv_{j}[(I-A)\fv_{0}(t)+\epsv
]_{j}\equiv\sum_{j=1}^{n}T_{j}.
\end{eqnarray*}
and apply Lindeberg's theorem (Theorem 1.15 in
\citet{s07}) to show its asymptotic distribution. First,
\begin{eqnarray}
Var(\sum_{j}T_{j})=\sum_{j}Var(T_{j})=\vv'\Gv_{n}\Gv_{n}'\vv(\sigma^{2}+n^{-1}\sum_{l=1}^{n}((I-A)\fv_{0})_{l}^{2})
\rightarrow\sigma^{2}\vv'\Gv\vv\label{inter5a}
\end{eqnarray}
by $\Gv_{n}\Gv_{n}'\rightarrow\Gv$ and (\ref{i-a}). We next verify
the condition that
\begin{eqnarray*}
\sum_{j=1}^{n}E(T_{j}^{2}I\{|T_{j}|>\delta\sigma\sqrt{\vv'\Gv\vv}\})
=o(\sigma^{2}\vv'\Gv\vv)
\end{eqnarray*}
for any $\delta>0$. Note that
\begin{eqnarray*}
\sum_{j=1}^{n}E(T_{j}^{2}I\{|T_{j}|>\delta\sigma\sqrt{\vv'\Gv\vv}\})&\leq&\sum_{j=1}^{n}(ET_{j}^{4})^{1/2}
(P(|T_{j}|>\delta\sigma\sqrt{\vv'\Gv\vv}))^{1/2}\\
&\leq&\left(\sum_{j=1}^{n}ET_{j}^{4}\right)^{1/2}\left(\sum_{j=1}^{n}P(|T_{j}|>\delta\sigma\sqrt{\vv'\Gv\vv})
\right)^{1/2}.
\end{eqnarray*}
In view of (\ref{inter5a}), we obtain
\begin{eqnarray*}
\sum_{j=1}^{n}P(|T_{j}|>\delta\sigma\sqrt{\vv'\Gv\vv})\leq\frac{\sum_{j=1}^{n}ET_{j}^{2}}{\delta^{2}\sigma^{2}\vv'G\vv}
\rightarrow\frac{1}{\delta^{2}}
\end{eqnarray*}
and
\begin{eqnarray*}
\sum_{j=1}^{n}ET_{j}^{4}&\leq&\frac{\|\vv\|^{4}\sum_{j=1}^{n}
E\|\Gv_{n}\bR_{11}^{-1/2}\wv_{j}\|^{4}E[(I-A)\fv_{0}+\epsv]
^{4}_{j}} {n^{2}}\\
&\leq&\frac{8\|\vv\|^{4}\sum_{j=1}^{n}E\|\Gv_{n}\bR_{11}^{-1/2}\wv_{j}\|^{4}([(I-A)\fv_{0}]_{j}^{4}+E\epsilon^{4})}{n^{2}}.
\end{eqnarray*}
Note that
\begin{eqnarray*}
E\|\Gv_{n}\bR_{11}^{-1/2}\wv_{j}\|^{4}\leq
lq_{n}^{2}\lambda_{min}^{-2}(\bR_{11})\sum_{i=1}^{l}\|g_{i}\|^{4}=O(q_{n}^{2}),
\end{eqnarray*}
where $\Gv_{n}'=(g_{1},\ldots,g_{l})$, due to
$\Gv_{n}\Gv_{n}'\rightarrow\Gv$. Combined with the above analysis
we have $\sum_{j}ET_{j}^{4}=O(q_{n}^{2}\lambda_{1}^{2}\vee
q_{n}^{2}n^{-1})$ given the sub-exponential tail of $\epsilon$ and
(\ref{i-a}). By the conditions that $q_{n}\leq d_{n}=o(n^{1/3})$
and $\lambda_{1}/q_{n}\asymp n^{-2m/(2m+1)}$, we have verified the
condition that
$\sum_{j=1}^{n}E(T_{j}^{2}I\{|T_{j}|>\delta\sigma\sqrt
{\vv'\Gv\vv}\})=o(\sigma^{2}\vv'\Gv\vv)$. Therefore, we have
proved that $(\ref{interdis})=N(0,\sigma^{2}\Gv)+o_{P}(1)$.

Then we have
\begin{eqnarray}
\sqrt{n}\Gv_{n}\bR_{11}^{1/2}(\widehat{\btv}_{1}-\btv_{10})&=&\sqrt{n}\Gv_{n}\bR_{11}^{-1/2}
\b(\bR_{11}-\Xv_{1}'(I-A)\Xv_{1}/n)(\widehat{\btv}_{1}-\btv_{10})+
N(0,\sigma^{2}\Gv)+o_{P}(1)\nonumber\\
&=&N(0,\sigma^{2}\Gv)+o(1)+O_{P}(d_{n}^{3/2}n^{-1/2}\vee
d_{n}^{3/2}n^{-1}\lambda_{1}^{-1/2m})\label{final}
\end{eqnarray}
by the matrix inequality, (\ref{cova}) and (\ref{pararatei}). The
stated condition $d_n=o(n^{1/3}\wedge n^{2/3}\lambda_{1}^{1/3m})$
implies that the rest term in (\ref{final}) is $o_{P}(1)$. This
completes the proof of (\ref{asydisti}). \end{proof}


\vskip -1cm

\begin{thebibliography}{1}
\bibliographystyle{plainnat}
\bibitem[Abramowitz and Stegun (1964)]{as64}
{\sc Abramowitz, M. and Stegun, I.} (1964) {\it Handbook of
Mathematical Functions with Formulas, Graphs, and Mathematical
Tables}. Dover, New York.


\bibitem[Breiman (1995)]{b95}
{\sc Breiman, L.} (1995). Better subset selection using the
nonnegative garrote. {\it Technometrics} {\bf 37}, 373-384.




\bibitem[Bickel et al.(2009)]{BRT09} {\sc Bickel, P. J., Ritov, Y. and Tsybakov, A. B.} (2009).
Simultaneous analysis of Lasso and Dantzig selector.
\textit{Annals of Statistics},
\textbf{37}, 1705--1732.

\bibitem[Craven and Wahba(1979)]{cw79}
{\sc Craven, P. and Wahba, G.} (1979). Smoothing noisy data with
spline functions: Estimating the correct degree of smoothing by
the method of generalized cross-validation. {\it Numerishe
Mathematik} {\bf 31}, 377-403.

\bibitem[Denby(1984)]{d84}
{\sc Denby, L.} (1984). Smooth regression functions. Ph.D. Thesis.
Department of Statistics. University of Michigan.

\bibitem[Efron et al.(2004)]{ehjt04}
{\sc Efron, B., Hastie, T., Johnstone, I. and Tibshirani, R.}
(2004). Least angle regression. {\it Annals of Statistics} {\bf
32}, 407-451.


\bibitem[Fan and Li(2001)]{fl01}
{\sc Fan, J. and Li, R.} (2001). Variable Selection via Nonconcave
Penalized Likelihood and Its Oracle Properties. {\it Journal of
American Statistical Association} {\bf 96}, 1348-1360.

\bibitem[Fan and Li(2004)]{fl04}
{\sc Fan, J. and Li, R.} (2004). New Estimation and Model
Selection Procedures for Semiparametric Modeling in Longitudinal
Data Analysis. {\it Journal of American Statistical Association}
{\bf 99}, 710-723.

\bibitem[Fan and Lv(2008)]{fl08}
{\sc Fan, J. and Lv, J.} (2008). Sure independence screening for
ultra-high dimensional feature space. (with discussion) {\it
Journal of Royal Statistical Society B} {\bf 70}, 849-911.

\bibitem[Fan and Peng(2004)]{fp04}
{\sc Fan, J. and Peng, H.} (2004). Nonconcave penalized likelihood
with a diverging number of parameters. {\it Annals of Statistics}
{\bf 32}, 928-961.



\bibitem[Green and Silverman(1994)]{gs94}
{\sc Green, P.J. and Silverman, B.W.} (1994). {\it Nonparametric
regression and generalized linear models}, London: Chapman and
Hall.


\bibitem[Gu(2002)]{g02}
{\sc Gu, C.} (2002). \emph{Smoothing spline ANOVA models.} New York: Springer-Verlag.
%

\bibitem[Heckman(1986)]{h86}
{\sc Heckman, N.} (1986). Spline smoothing in a partly linear
models. {\it Journal of Royal Statistical Society, Series B} {\bf
48}, 244-248.

\bibitem[Huang et al.(2008a)]{hhm08}
{\sc Huang, J., Horowitz, J. and Ma, S.} (2008a), Asymptotic
properties of bridge estimators in sparse high-dimensional
regression models, {\it Annals of Statistics} {\bf 36}, 587-613.

\bibitem[Huang et al.(2008b)]{hmz08}
{\sc Huang, J., Ma, S., and Zhang, C. H.} (2008b), Adaptive LASSO
for sparse high dimensional regression, {\it Statistica Sinica}
{\bf 18}, 1603-1618.

\bibitem[Kimeldorf and Wahba(1971)]{kw71}
{\sc Kimeldorf, G and Wahba, G.} (1971). Some results on
Tchebycheffian spline functions. {\it Journal of Mathematical Analysis and Applications} {\bf 33}, 82-95.

%
%

\bibitem[Mammen and van de Geer(1997)]{mvg97}
{\sc Mammen, E. and van de Geer, S.} (1997). Penalized
quasi-likelihood estimation in partially linear models. {\it
Annals of Statistics} {\bf 25}, 1014-1035.

\bibitem[Ni et al.(2009)]{NZZ09} {\sc Ni, X., Zhang, H. H. and Zhang, D.} (2009).
Automatic Model Selection for Partially Linear Models. \emph{Journal of Multivariate Analysis}
\textbf{100}, 2100--2111.


\bibitem[Portnoy(1984)]{p84}
{\sc Portnoy, S.} (1984). Asymptotic Behavior of M-Estimator of
$p$ Regression Parameters when $p^2/n$ is large. I. Consistency.
{\it Annals of Statistics} {\bf 12}, 1298-1309.


\bibitem[Rice(1986)]{r86}
{\sc Rice, J.} (1986). Convergence Rates for Partially Spline
Model. {\it Statistics and Probability Letters} {\bf 4}, 203-208.
%
%

\bibitem[Ruppert(2003)]{r03}
{\sc Ruppert, D., Wand, M.P. and Carroll, R.J.} (2003)
\emph{Semiparametric Regression}. Cambridge University Press.

\bibitem[Shang and Cheng(2013)]{SC13}
{\sc Shang, Z. and Cheng, G.} (2013) Local and Global Asymptotic Inference in Smoothing Spline Models. {\it Annals of Statistics}, To Appear.

\bibitem[Shao(2003)]{s07}
{\sc Shao, J.} (2003) \emph{Mathematical Statistics.} 2nd Ed, Springer. New York.

\bibitem[Shiau and Wahba(1988)]{sw88}
{\sc Shiau, J. and Wahba, G.} (1988). Rates of convergence for
some estimates of a semi-parametric model. {\it Communications in Statistics - Simulation and Computation} {\bf 17}, 111-113.

\bibitem[Speckman(1988)]{s88}
{\sc Speckman, P.} (1988). Kernel smoothing in partially linear
models. {\it Journal of Royal Statistical Society-B} {\bf 50},
413-436.

\bibitem[Stamey et al.(1989)]{s89}
{\sc Stamey, T., Kabalin, J., McNeal, J., Johnstone, I., Freida, F., Redwine, E., and Yang, N.} (1989).
Prostate specific antigen in the diagnosis and treatment of adenocarcinoma of the prostate II radical
prostatectomy treated patients. {\it Journal of Urology} {\bf 16}, 1076-1083.

\bibitem[Tibshirani(1996)]{t96}
{\sc Tibshirani, R.} (1996). Regression shrinkage and selection
via the lasso. {\it Journal of the Royal Statistical Society,
Series B} {\bf 58}, 267-288.

\bibitem[van der Vaart and Wellner(1996)]{vw96}
{\sc van der Vaart, A. W., and Wellner, J. A.} (1996) \emph{Weak Convergence
and Empirical Processes: With Applications to Statistics.}
Springer, New York

\bibitem[Wahba(1984)]{w84}
{\sc Wahba, G.} (1984) Partial spline models for the
semiparametric estimation functions of several variables. In {\it
Statistics: An Appraisal, Proceedings of the 50th Anniversary Conference}, eds H. A.
David and H. T. David. Ames: Iowa State University Press.

\bibitem[Wahba(1990)]{w90}
{\sc Wahba, G.} (1990) {\it Spline Models for Observational Data}.
SIAM. CBMS-NSF Regional Conference Series in Applied Mathematics, volume 59.
Philadelphia.

\bibitem[Wang et al.(2007a)]{wlt07}
{\sc Wang, H., Li, R., and Tsai, C.L.} (2007a). Tuning parameter
selectors for the smoothly clipped absolute deviation method. {\it
Biometrika} {\bf 94}, 553-568.

\bibitem[Wang et al.(2007b)]{wlj07}
{\sc Wang, H., Li, G., and Jiang, G.} (2007b).  Robust regression
shrinkage and consistent variable selection via the LAD-LASSO.
{\it Journal of Business \& Economics Statistics} {\bf 20},
347-355.

\bibitem[Wang et al.(2009)]{wll09}
{\sc Wang, H., Li, B., and Leng, C.} (2009) Shrinkage tuning
parameter selection with a diverging number of parameters. {\it
Journal of Royal Statistical Society, Series B} {\bf 71} 671-683.

\bibitem[Yatchew(1997)]{y97}
{\sc Yatchew, A.} (1997). An elementary estimator of the partial
linear model. {\it Economics Letters} {\bf 57}, 135-143.

\bibitem[Zhang and Lu(2007)]{zl07}
{\sc Zhang, H. H. and Lu, W.} (2007). Adaptive-LASSO for Cox's
proportional hazards model. {\it Biometrika} {\bf 94}, 691-703.

\bibitem[Zou(2006)]{z06}
{\sc Zou, H.} (2006). The adaptive lasso and its oracle
properties. {\it Journal of American Statistical Association} {\bf 101}, 1418-1429.

\bibitem[Zou(2009)]{zz08}
{\sc Zou, H. and Zhang, H. H.} (2009). On The Adaptive Elastic-Net
With A Diverging Number of Parameters.  {\it Annals of Statistics}
{\bf 37}, 1733-1751.
\end{thebibliography}
\end{document}